\documentclass[journal,onecolumn]{IEEEtran}

\usepackage{cite}
\usepackage{amsthm,amsmath,amssymb}
\usepackage{mathrsfs}
\usepackage{multirow}
\usepackage{amsthm}
\usepackage{microtype}
\usepackage{hyperref}
\usepackage[capitalize]{cleveref}
\usepackage{verbatim}
\usepackage{xcolor}
\usepackage{enumitem}
\usepackage{arydshln} 

\newtheorem{theorem}{Theorem}[section]
\newtheorem{lemma}[theorem]{Lemma}

\newtheorem{corollary}[theorem]{Corollary}

\newtheorem{definition}[theorem]{Definition}

\newtheorem{remark}[theorem]{Remark}
\newtheorem{claim}[theorem]{Claim}
\newtheorem{example}[theorem]{Example}

\begin{document}

\title{A Disguise-and-Squeeze PIR Scheme for the MDS-TPIR Setting and Beyond}
\author{{\textbf{Rui Sun}, \textbf{Ran Tao}, \textbf{Jingke Xu}, and \textbf{Yiwei Zhang},~\IEEEmembership{Member,~IEEE}}
\thanks{R. Sun, R. Tao, and Y. Zhang are with State Key Laboratory of Cryptography and Digital Economy Security, Key Laboratory of Cryptologic Technology and Information Security of Ministry of Education, School of Cyber Science and Technology, Shandong University, Qingdao, Shandong, 266237, China.  J. Xu is with School of Information Science and Engineering, Shandong Agricultural University, Tai'an, Shangdong, 271018, China.}
\thanks{Research supported in part by National Key Research and Development Program of China under Grant Nos. 2022YFA1004900 and 2021YFA1001000, in part by National Natural Science Foundation of China under Grant Nos. 12231014, 12401437, and 12201362, in part by Shandong Provincial Natural Science Foundation under Grant Nos. ZR2022QA069 and ZR2025MS59, and in part by the Taishan Scholars Program. Part of this paper~\cite{sun2025mds} has been presented at IEEE International Symposium on Information Theory (ISIT), 2025. Corresponding author: Y. Zhang (ywzhang@sdu.edu.cn).}
}

\maketitle

\begin{abstract}
We consider the problem of private information retrieval (PIR) from MDS coded databases with colluding servers, i.e., MDS-TPIR. In the MDS-TPIR setting, $M$ files are stored across $N$ servers, where each file is stored independently using an $(N,K)$-MDS code. A user wants to retrieve one file without disclosing the index of the desired file to any set of up to $T$ colluding servers. The general problem in studying PIR schemes is to maximize the PIR rate, defined as the ratio of the size of the desired file to the size of the total download. Freij-Hollanti et al. proposed a conjecture of the MDS-TPIR capacity (the maximum achievable PIR rate), which was later disproved by Sun and Jafar by a counterexample with $(M,N,T,K)=(2,4,2,2)$. 

In this paper, we propose a new MDS-TPIR scheme based on a disguise-and-squeeze approach. The features of our scheme include the following.
\begin{itemize}
    \item Our scheme generalizes the Sun-Jafar counterexample to $(M,N,T,K)=(2,N,2,K)$ with $N\geq K+2$ for an arbitrary $(N,K)$-MDS coded system, providing more counterexamples to the conjecture by Freij-Hollanti et al.
    \item For $(M,N,T,K)=(2,N,2,K)$ and a GRS (generalized Reed-Solomon codes) coded system, our scheme has rate $\frac{N^2-N}{N^2+KN-2K}$, beating the state-of-the-art results. We further show that this rate achieves the linear MDS-TPIR capacity when $K=2$.
    \item Our scheme features a significantly smaller field size for implementation and the adaptiveness to generalized PIR models such as multi-file MDS-TPIR and MDS-PIR against cyclically adjacent colluding servers.
    \item Lastly, we provide an $\epsilon$-error MDS-TPIR scheme for $T\geq 3$ based on the disguise-and-squeeze framework.
\end{itemize}

\end{abstract}

\section{Introduction}
In the digital age, data retrieval often comes with the risk of privacy breaches, especially when personal preferences are used for targeted services. This is particularly critical when accessing sensitive databases, such as those in healthcare or finance. {\it Private information retrieval} (PIR) has emerged to allow users to query databases while preserving the privacy of their search intentions. The concept of PIR was first introduced by Chor et al.~\cite{492461, chor1998private}, focusing on privately retrieving one out of multiple bits. The primary metric in this case is the communication complexity, defined as the total number of bits exchanged between the user and the servers to retrieve a single bit, including both the upload of the queries by the user and the download of the responses from servers. In recent years, the PIR problem in {\it distributed storage systems}~\cite{dimakis2010network} has become an important area of research. In this context, it is typically assumed that the user wants to retrieve one out of a set of files stored in the system, while the size of each file can be arbitrarily large. For the PIR schemes in distributed storage systems, the download is always proportional to the size of the file whereas the upload cost can be much smaller and thus often neglected~\cite{chan2015private}. Consequently, the optimization focuses on minimizing the download cost only. For a comprehensive overview of PIR and related topics, see the survey articles~\cite{gasarch2004survey,ulukus2022private,vithana2023private,d2024guided} and the references therein. 

In particular, PIR from MDS coded databases with colluding servers, abbreviated as {\it MDS-TPIR}, has gained considerable attention in recent years. In the setting of MDS-TPIR, $M$ files are stored across $N$ servers, where each file is stored independently using the same $(N,K)$-MDS code. A user wants to retrieve a single file without disclosing the index of the desired file to any set of up to $T$ colluding servers. The {\it PIR rate} of an MDS-TPIR scheme is defined as the ratio of the number of bits of the desired file to the total bits downloaded from all servers. The maximum achievable rate is called the {\it MDS-TPIR capacity}. In a series of works~\cite{sun_12017capacity,sun_22017capacity,banawan2018capacity}, the MDS-TPIR capacity has been determined for the degenerated cases with $K=1$ or $T=1$. These schemes rely on three key principles posed in~\cite{sun_12017capacity}: symmetry across files, symmetry across databases, and the use of side information from undesired files through interference alignment. The MDS-TPIR capacity has been formulated as 
$$C=(1+\rho+\dots+\rho^{M-1})^{-1},$$
where $\rho$ equals $\frac{1}{N}$, $\frac{K}{N}$, and $\frac{T}{N}$ for the three cases $K=T=1$ (PIR, for repetition-coded system with no collusion), $T=1$ (MDS-PIR, for $(N,K)$-MDS coded system with no collusion), and $K=1$ (TPIR, for repetition-coded system with $T$-collusion), accordingly. Later, several papers~\cite{tian2019capacity,xu2018sub,zhang2018optimal,xu2018building,xu2019capacity,XWTCOM22,XFTIT25,zhu2019new,zhou2020capacity,zhu2021capacity} further modified these capacity-achieving schemes to minimize the sub-packetization level or the size of the filed.

For the non-degenerated cases $K\geq 2$ and $T\geq 2$, the MDS-TPIR capacity is far from being fully solved. Zhang and Ge~\cite{zhang2019general} proposed a scheme with rate $(1+\rho+\dots+\rho^{M-1})^{-1}$, where $\rho = 1 - {\binom{N-T}{K}}\cdot {\binom{N}{K}}^{-1}$. In particular, this result captures the essence of the previous capacity-achieving degenerated cases in~\cite{sun_12017capacity,sun_22017capacity,banawan2018capacity}. Using the Schur product of generalized Reed-Solomon (GRS) codes, Freij-Hollanti et al.~\cite{freij2017private} proposed a PIR scheme (for GRS-coded databases) with rate $\frac{N - (K + T - 1)}{N}$. Combining their results with the known capacity expression for the degenerated cases, Freij-Hollanti et al.~\cite{freij2017private}  further conjectured that the MDS-TPIR capacity is of the form 
\begin{equation}\label{conj}
C_{\mathrm{MDS-TPIR}}^\mathrm{conj}= \left(1 + \rho +  \cdots + \rho^{M-1}\right)^{-1}, ~  \rho = \frac{K + T - 1}{N}.
\end{equation}

The conjecture, referred to as the {\it FGHK conjecture} in the rest of the paper, is appealing since it also contains all the known degenerated cases as special cases. Towards proving the FGHK conjecture, MDS-TPIR schemes (for GRS-coded databases) with rate exactly the same as Equation (\ref{conj}) have been proposed in~\cite{d2019one} and~\cite{holzbaur2021toward}, through a refinement and lifting approach. These schemes are also the state-of-the-art results for most parameters. Moreover, the conjecture is proven to be correct when the PIR scheme is further required to be linear and full support-rank~\cite{holzbaur2021toward}. However, if there is no such restriction, the FGHK conjecture has been disproved by Sun and Jafar~\cite{sun2017private}, who constructed a PIR scheme for $(M, N, T, K) = (2, 4, 2, 2)$, achieving a PIR rate of $3/5$ exceeding the conjectured value of $4/7$. Furthermore, for the specific case with $K = N - 1$ and $M = 2$, Sun and Jafar~\cite{sun2017private} provided an exact expression for the capacity as $\frac{N^2 - N}{2N^2 - 3N + T}$, and pointed out that the roles of $K$ and $T$ are not necessarily symmetric in the MDS-TPIR capacity expression, unlike what is expected based on previously known results. 

Recently there have been many extensions of PIR, such as: PIR with unresponsive or robust servers~\cite{sun_22017capacity,tajeddine2017robust}; PIR with Byzantine servers~\cite{Byzantinebanawan2018capacity,tajeddine2019private}; symmetric PIR~\cite{symmetricsun2018capacity,wang2017secure,wang2019symmetric,holzbaur2019capacitysymadver}; PIR in non-MDS-coded system~\cite{kumar2019achievingnonmds,lin2018nonmds};  PIR in joint storage coded systems~\cite{sun2019breaking}; PIR with side information or caches~\cite{kadhe2019private,tandon2017capacity,wei2018fundamental,chen2020capacity,wei2019capacity}; and multiround PIR~\cite{sun2018multiround,yao2019capacity}, etc. In particular, in this paper we will come across the following two variant models. The first one is the {\it multi-file PIR}, where a user needs to retrieve $P\geq 2$ files out of $M$ files simultaneously. Banawan and Ulukus \cite{banawan2018multi} first proposed multi-file PIR schemes which are more efficient than repeatedly performing $P$ single-file PIR schemes, and they determined the exact capacity when $K=T=1$ and $P\geq M/2$. Zhang and Ge~\cite{zhang2017private} proposed a multi-file MDS-TPIR scheme for $P\geq M/2$, achieving the capacity when either $K=1$ or $T=1$. The second one is {\it PIR with restricted collusion patterns}. This problem was initially introduced by Tajeddine et al.~\cite{arbitrarytajeddine2017private} where the collusion pattern is characterized by a given family of maximal colluding sets. For the repetition-coded system, the capacity has been determined for disjoint colluding sets~\cite{jia2017capacitydisjoint} and then later for arbitrary collusion patterns by a linear programming method~\cite{arbitraryyao2021capacity}. Note that in their paper~\cite{sun2017private} disproving the FGHK conjecture, Sun and Jafar also provided several values of the PIR capacity under restricted collusion patterns in MDS-coded storage systems. 

In this paper, by deeply analyzing the $(2,4,2,2)$-counterexample of Sun and Jafar, we propose a new MDS-TPIR scheme using a {\it disguise-and-squeeze} approach. The disguising phase is for the privacy, in which we carefully design the queries regarding every file such that a set of $T$ colluding servers cannot tell the difference between the queries for desired and undesired files. In the squeezing phase, we analyze how to make use of the redundancy among the queried undesired symbols, such that each server can apply  a {\it combination strategy} when responding to the queries so as to improve the PIR rate. The better strategy we design the squeezing phase, the larger PIR rate we could derive. 

Our main contributions, as well as the structure of this paper, are summarized as follows:
\begin{itemize}
    \item We introduce the basic notations and terminologies, and briefly recall the Sun-Jafar counterexample in Section~\ref{Sec:pre}. 
    \item Our scheme generalizes the Sun-Jafar counterexample to $(M,N,T,K)=(2,N,2,K)$ with $N\geq K+2$ for an arbitrary $(N,K)$-MDS coded system. The rate of our scheme is $R=\frac{N^2-N}{2N^2-2N+K^2-NK}$ when $N\leq 2K$ and $R=\frac{N^2-N}{N^2-N+2NK-K^2-K}$ when $N>2K$. These results provide more counterexamples to the FGHK conjecture (Section \ref{Sec:main}).
    \item By using a combination strategy different from~\cite{sun2017private}, the size of the finite field to implement the disguise-and-squeeze MDS-TPIR scheme is dramatically reduced. (Subsection \ref{subsec:field})
    \item In particular, when $(M,N,T,K)=(2,N,2,K)$, $N\geq K+2$, and the system is based on a GRS code, the rate of our scheme can be further improved to $R=\frac{N^2-N}{N^2+KN-2K}$, beating the state-of-the-art results in~\cite{d2019one} and~\cite{holzbaur2021toward} (Section \ref{Sec:squeeze}). Moreover, when $K=2$, we prove that our scheme for GRS-coded system achieves the linear MDS-TPIR capacity (Subsection \ref{subsec:con-main}). 
    \item As by-products, our scheme can also be adapted to some extension models of PIR. For the multi-file MDS-TPIR, we beat the results from~\cite{zhang2017private} for non-degenerated cases for a certain range of parameters (Subsection \ref{subsec:multi}). For PIR with restricted collusion patterns, we 
    establish the capacity for the case $(M, N, K) = (2, N, K)$ when any two cyclically adjacent servers can collude (Subsection \ref{subsec:cyc} and Subsection \ref{subsec:con-cyc}).
    \item Finally, we propose a disguise-and-squeeze $\epsilon$-error MDS-TPIR scheme for $T\geq 3$, where we make use of the tool of exterior products to design the queries (Section \ref{Sec:Tgeq3}).
\end{itemize}

\section{Preliminaries}\label{Sec:pre}

For $a, b\in \mathbb{Z}^{+}$, define $[a: b]$ as the set $\{a, a+1,\dots, b\}$, and $[1:n]$ is abbreviated as $[n]$.  For an index set $\mathcal{I} = \{i_1, \dots, i_{k}\}\subseteq[n]$ with $i_1 < \cdots < i_{k}$, and for any vector $\mathbf{x}=(x_1,\dots,x_n)$, let $\mathbf{x}_{\mathcal{I}}$ be the subvector $(x_{i_1}, \dots, x_{i_{k}})$. The notation $X \sim Y$ is used to indicate that two random variables $X$ and $Y$ are identically distributed. 

\subsection{Problem setting of MDS-TPIR}

Consider $M$ independent files $W_1, \dots, W_M \in \mathbb{F}_q^{L \times 1}$. In $q$-ary units,
\begin{align}
    H(W_1) &= \dots = H(W_M) = L,\label{W}  \\
    H(W_1, \dots, W_M) &= H(W_1) + \dots + H(W_M). 
\end{align}

The files are stored in a system consisting of $N$ servers. 
Denote $W_{m,n} \in \mathbb{F}_q^{\frac{L}{K} \times 1}$ as the 
$\frac{L}{K}$ coded symbols of the file $W_m$ stored in the $n^{\mathrm{th}}$ server, for $m \in [M]$ and $n\in[N]$. When there is no ambiguity, $W_{m,n}$ is abbreviated as $W_{mn}$. 
The $n^{\mathrm{th}}$ server stores $(W_{1n}, W_{2n}, \dots, W_{Mn})$. The storage system applies the same $(N,K)$-MDS code independently to each file. Thus the following hold:
\begin{align}
    H(W_{mn} | W_m) &= 0, \quad H(W_{mn}) = L/K,\label{W_{mn}} \\
 [\mbox{MDS}] ~    H(W_m | W_{m\mathcal{K}}) &= 0, \quad \forall \mathcal{K}\subseteq [N], \; |\mathcal{K}| = K. \label{mds}
\end{align}
Here $W_{m\mathcal{K}}$ denotes the set $\{ W_{mn}:n\in\mathcal{K} \}$. Let $\mathcal{F}$ be a random variable generated privately by the user, whose realization is
not available to the servers. The variable $\mathcal{F}$ captures the randomness in the strategies employed by the user to design the queries. Similarly, let $\mathcal{G}$ be a random variable representing the random strategies adopted by the servers to produce the responses, with its realizations assumed to be known by both the servers and the user. The user privately and uniformly generates a random variable $\theta \in [1:M]$ and wishes to retrieve $W_{\theta}$, while keeping $\theta$ secret from any up to $T$ colluding servers. Here $\mathcal{F}$ and $\mathcal{G}$ are generated independently and before the realizations of the files or knowing the desired file index. Hence, the following equality holds:
\begin{align}
H(\theta, \mathcal{F}, \mathcal{G}, W_1, \ldots, W_M) &= H(\theta) + H(\mathcal{F}) + H(\mathcal{G}) + \sum_{i=1}^M H(W_i).\label{F-W}
\end{align}

Suppose the user picks $\theta=m$. In order to retrieve $W_m$ privately, the user privately generates $N$ queries $Q_1^{[m]}, \dots, Q_N^{[m]}$.
\begin{align}
H(Q_1^{[m]},\dots,Q_N^{[m]}|\mathcal{F})=0, \quad \forall m\in [M].\label{F-Q}
\end{align}

The user sends the query $Q_n^{[m]}$ to the $n^{\mathrm{th}}$ server, for $n \in [N]$. Upon receiving $Q_n^{[m]}$, the $n^{\mathrm{th}}$ server responds with $A_n^{[m]}$, which is a deterministic function of $Q_n^{[m]}$, $W_{1n}, W_{2n}, \dots, W_{Mn}$ and $\mathcal{G}$, i.e.,
 \begin{align}\label{Answer}
    H(A_n^{[m]} | Q_n^{[m]}, W_{1n}, \dots, W_{Mn},\mathcal{G}) &= 0. 
\end{align}

Based on the feedback $A_1^{[m]},\dots,A_N^{[m]}$, the user can retrieve the desired file $W_m$,
\begin{align}\label{correctness}
 [\mbox{Correctness}] ~    H(W_m | A_{1:N}^{[m]},Q_{1:N}^{[m]},\mathcal{F},\mathcal{G}) &= 0.
\end{align}

To satisfy the $T$-privacy requirement, it must hold that
\begin{align}\label{privacy}
 [\mbox{Privacy}] ~    I(Q_{\mathcal{T}}^{[m]};m)=0, \quad \forall \mathcal{T}\subseteq [N],\; |\mathcal{T}|=T.
\end{align}

The PIR rate for the scheme is defined as the ratio of the size of the retrieved file to the size of total download, i.e.,
\begin{align*}
 R=\frac{ L}{\sum_{n=1}^{N}H(A_n^{[m]})}.
\end{align*}

The PIR capacity $C$ is the maximum achievable PIR rate among all available schemes. Sometimes, to emphasize the parameters, the rate and the capacity are written as $R(M,N,T,K)$ and $C(M,N,T,K)$, respectively.

\subsection{A brief review of the counterexample by Sun and Jafar}

For $(M,N,T,K)=(2,4,2,2)$, the FGHK-conjectured capacity is $4/7$. Sun and Jafar disproved the conjecture by providing a counterexample. We briefly review the counterexample for self-completeness of the paper, with some terminologies changed according to our general framework.
The readers are also expected to refer to this example when reading our main results later.

\begin{example}\label{emp:SunJafar}
[Sun and Jafar~\cite{sun2017private}] Let $\mathbb{F}_q$ be a sufficiently large finite field. Let the two files be $W_1=(\mathbf{a},\mathbf{b})$ and $W_2=(\mathbf{c},\mathbf{d})$, where $\mathbf{a},\mathbf{b},\mathbf{c},\mathbf{d}$ are vectors in $\mathbb{F}_q^{6\times1}$. 
The storage is specified as $(W_{11},W_{12},W_{13},W_{14})=(\mathbf{a},\mathbf{b},\mathbf{a}+\mathbf{b},\mathbf{a}+2\mathbf{b})$ and $(W_{21},W_{22},W_{23},W_{24})=(\mathbf{c},\mathbf{d},\mathbf{c}+\mathbf{d},\mathbf{c}+2\mathbf{d})$. 

\subsubsection{The disguising phase}
The user picks two matrices $S$ and $S'$ from the set of all full rank $6\times 6$ matrices over $\mathbb{F}_q$, privately and independently. Label the rows of $S$ as $\{V_1,\dots,V_6\}$. Label the rows of $S'$ as $\{Z_1,U_1,\dots,U_5\}$. Set the following query sets:
\begin{align*}
\mathcal{V}_1&=\{V_1,V_2,V_3\}, \quad \mathcal{U}_1=\{Z_1,\widetilde{U}_{1,1},\widetilde{U}_{1,2}\}, \\   
\mathcal{V}_2&=\{V_1,V_4,V_5\}, \quad \mathcal{U}_2=\{Z_1,\widetilde{U}_{2,1},\widetilde{U}_{2,2}\}, \\
\mathcal{V}_3&=\{V_2,V_4,V_6\}, \quad \mathcal{U}_3=\{Z_1,\widetilde{U}_{3,1},\widetilde{U}_{3,2}\}, \\
\mathcal{V}_4&=\{V_3,V_5,V_6\}, \quad \mathcal{U}_4=\{Z_1,\widetilde{U}_{4,1},\widetilde{U}_{4,2}\}, 
\end{align*}
where  
\begin{eqnarray*}
\begin{array}{l}
\left(\begin{array}{cc}\widetilde{U}_{1,1} & \widetilde{U}_{1,2}  \\ 
 
\vdots &\vdots \\
\widetilde{U}_{4, 1} & \widetilde{U}_{4,2} \end{array}\right)
=\left(\begin{array}{cc}
        1 &1  \\
         1&2\\
       1&0\\
       0&1
\end{array}\right)
\left(\begin{array}{ll}U_{1} & U_{3} \\
U_{2} & U_4\end{array}\right).
\end{array} 
\end{eqnarray*} 

Denote the coefficient matrix in the equation above as ${\bf H}$. Note that ${\bf H}$ is row-MDS.
Here $\mathcal{V}_n$ is the query set regarding the desired file $W_{1n}$ and $\mathcal{U}_n$ is the query set regarding the undesired file $W_{2n}$, for the $n^{\mathrm{th}}$ server. The privacy is guaranteed by further applying random and independent permutations $\pi_n\in S_3$ to $\mathcal{V}_n$ and $\pi_n'\in S_3$ to $\mathcal{U}_n$, for $n\in[4]$, such that any two colluding servers can only observe that their query sets regarding each file have exactly one common vector, and thereby altogether 5 linearly independent vectors. After receiving the queries, the $n^{\mathrm{th}}$ server produces the queried desired symbols $\pi_n(\mathcal{V}_n)W_{1n}$ and the queried undesired symbols $\pi'_n(\mathcal{U}_n)W_{2n}$.

\subsubsection{The squeezing phase}
The improvement in the PIR rate comes from the redundancy among the queried undesired symbols. There are two redundant symbols among $\{Z_1W_{2n}:1\leq n\leq 4\}$ since the storage system is $(4,2)$-MDS. By the property of the matrix ${\bf H}$, we have
\begin{align*}
  (U_1+U_2)\mathbf{c}+(U_1+2U_2)\mathbf{d}&=U_1(\mathbf{c}+\mathbf{d})+U_2(\mathbf{c}+2\mathbf{d}), \\
  (U_3+U_4)\mathbf{c}+(U_3+2U_4)\mathbf{d}&=U_3(\mathbf{c}+\mathbf{d})+U_4(\mathbf{c}+2\mathbf{d}).
\end{align*}
Thus, there is one redundant symbol among $\{\widetilde{U}_{1,1}\mathbf{c},\widetilde{U}_{2,1}\mathbf{d},\widetilde{U}_{3,1}(\mathbf{c}+\mathbf{d}),\widetilde{U}_{4,1}(\mathbf{c}+2\mathbf{d})\}$, and one among $\{\widetilde{U}_{1,2}\mathbf{c},\widetilde{U}_{2,2}\mathbf{d}, \widetilde{U}_{3,2}(\mathbf{c}+\mathbf{d}),\widetilde{U}_{4,2}(\mathbf{c}+2\mathbf{d})\}$. Therefore, $\{\mathbb{\pi}_n'(\mathcal{U}_n)W_{2n}: n\in [4]\}$ span a linear space of dimension at most $8$. 

Instead of answering the queries directly, each server applies a ``combination strategy'' to map its 6 queried symbols into 5 downloaded symbols. The five downloaded symbols from each server include 2 downloaded desired symbols, 2 downloaded undesired symbols, and another 1 paired-up summation of symbols from both files. The key of the combination strategy is to guarantee that the 8 downloaded undesired symbols, regardless of how the user permutes the queries in each query set, could always linearly generate all the queried undesired symbols $\{\mathbb{\pi}_n'(\mathcal{U}_n)W_{2n}: n\in [4]\}$. Then the interferences of the undesired file in the 4 paired-up summation symbols can be eliminated. The user gets altogether 12 linearly independent desired symbols and the retrieval is done. It should be emphasized that the combination strategy should be a deterministic mechanism arranged on the servers, without knowing the permutations on the query sets used by the user. Therefore, designing the combination strategy is not a trivial task and may require a large finite field to implement. In the Sun-Jafar example, they managed to find a combination strategy in $\mathbb{F}_{349}$. For more details, please refer to~\cite{sun2017private}.
 
Due to the combination strategy in downloading, the total number of downloaded symbols is $20$ and the rate is $3/5$.    
\end{example}

\section{A new MDS-TPIR scheme: disguise and squeeze}\label{Sec:main}

Our main result in this paper is a generalization of the example given by Sun and Jafar, which leads to a new MDS-TPIR scheme based on a disguise-and-squeeze approach. In this section, we will illustrate our idea for $(M,N,T,K)=(2,N,2,K)$, $N\geq K+2$. In the upcoming sections, we will explain that our scheme also works for arbitrary $M$ and $T$. We start with an example for $(M,N,T,K)=(2,5,2,3)$.

\subsection{An example for $(M,N,T,K)=(2,5,2,3)$}\label{(2,5,2,3)}
\subsubsection{Files and storage code}
Each file is comprised of $L=30$ independent symbols from the finite field $\mathbb{F}_7$:
\begin{align*}
    W_1=({\bf a}_1, {\bf a}_2,{\bf a}_3),\quad W_2=({\bf b}_1, {\bf b}_2,{\bf b}_3),
\end{align*}
where ${\bf a}_1, {\bf a}_2,{\bf a}_3, {\bf b}_1, {\bf b}_2, {\bf b}_3\in \mathbb{F}_7^{10\times1}$. Each file is stored in the system via a $(5,3)$-MDS code. Let 
\begin{align*}
    {\bf G}=\left(\begin{array}{ccccc}
        1 &0&0&1&1  \\
         0&1&0&1&2 \\
         0&0&1&1&3
    \end{array}\right),\; {\bf H}=\left(\begin{array}{cc}
      1   &1  \\
       1 &2\\
       1&3\\
       1&0\\
       0&1
    \end{array}\right).
\end{align*}
Both ${\bf G}$ and ${\bf H}$ are MDS.
The storage is specified as
\begin{equation*}
 (W_{11},\dots, W_{15})=W_1{\bf G}, \quad (W_{21},\dots, W_{25})=W_2{\bf G},   
\end{equation*}
where $W_{mn}\in \mathbb{F}_7^{10\times1}$ for $m\in[2],n\in[5]$.

\subsubsection{The disguising phase} 
Let $S$ and $S'$ be independently and uniformly chosen from the set of all full rank $10\times 10$ matrices over $\mathbb{F}_7$. Label the rows of $S$ as $\{V_i: i\in [10]\}$ and the rows of $S'$ as $\{Z_1,Z_2,Z_3, U_1, \dots, U_7\}$. Define
 \begin{align*}
    \mathcal{V}_1&=\{V_1,V_2,V_3,V_4,V_5,V_6\},\quad ~ \mathcal{U}_1=\{Z_1, Z_2, Z_3, \widetilde{U}_{1,1},\widetilde{U}_{1,2},\widetilde{U}_{1,3}\},
    \\
    \mathcal{V}_2&=\{V_1,V_2,V_3,V_7,V_8,V_9\},
    \quad ~ \mathcal{U}_2=\{Z_1, Z_2, Z_3, \widetilde{U}_{2,1},\widetilde{U}_{2,2},\widetilde{U}_{2,3}\},\\
    \mathcal{V}_3&=\{V_1,V_4,V_5,V_7,V_8,V_{10}\},\quad \mathcal{U}_3=\{Z_1, Z_2, Z_3, \widetilde{U}_{3,1},\widetilde{U}_{3,2},\widetilde{U}_{3,3}\},\\
    \mathcal{V}_4&=\{V_2,V_4,V_6,V_7,V_9,V_{10}\},\quad \mathcal{U}_4=\{Z_1, Z_2, Z_3, \widetilde{U}_{4,1},\widetilde{U}_{4,2},\widetilde{U}_{4,3}\},\\
    \mathcal{V}_5&=\{V_3,V_5,V_6,V_8,V_9,V_{10}\},\quad \mathcal{U}_5=\{Z_1, Z_2, Z_3, \widetilde{U}_{5,1},\widetilde{U}_{5,2},\widetilde{U}_{5,3}\},
\end{align*}
where 
   \begin{eqnarray*}
\begin{array}{l}
\left(\begin{array}{ccc}\widetilde{U}_{1,1} & \widetilde{U}_{1,2} & \widetilde{U}_{1, 3} \\ 
 
\vdots &\vdots & \vdots\\
\widetilde{U}_{5, 1} & \widetilde{U}_{5,2} & \widetilde{U}_{5, 3}\end{array}\right) 
\triangleq{\bf H}\left(\begin{array}{lll}U_{1} & U_2 & U_{3} \\
U_{4} & U_5 & U_{6}\end{array}\right).
\end{array}
\end{eqnarray*} 

Let $W_m$ be the desired file and $W_{m^c}$ be the undesired file.  
The query sent to the $n^{\mathrm{th}}$ server, for $n\in [5]$, contains two parts: $Q_n^{[m]} = (Q_n^{[m]}(W_m), Q_n^{[m]}(W_{m^c}))=(\mathbb{\pi}_n(\mathcal{V}_n),\mathbb{\pi}_n'(\mathcal{U}_n))$,
where $\mathbb{\pi}_n$ and $\mathbb{\pi}_n'$ are random permutations independently and uniformly chosen from the symmetric group $S_6$. Write $\mathbb{\pi}_n(\mathcal{V}_n)$ and $\mathbb{\pi}_n'(\mathcal{U}_n)$ in the form of $6\times 10$ matrices over $\mathbb{F}_7$, respectively. After receiving the queries, the $n^{\mathrm{th}}$ server applies the query vectors to its stored contents $(W_{mn},\, W_{m^c n})$, producing $6$ {\it queried desired symbols} $\mathbb{\pi}_n(\mathcal{V}_n)W_{mn}$ and $6$ {\it queried undesired symbols} $\mathbb{\pi}_n'(\mathcal{U}_n)W_{m^cn}$.

Now we briefly analyze how the queries keep the privacy by disguising. 
Here $\{\mathcal{V}_n:n\in[5]\}$ are the query sets for the desired file, where every $K=3$ sets share exactly one common vector. Consequently, every two colluding sets in $\{\mathcal{V}_n:n\in[5]\}$ share exactly $3$ common vectors, and this key property is imitated in the query sets $\{\mathcal{U}_n:n\in[5]\}$ for the undesired file, by the common vectors $Z_1,Z_2,Z_3$ appearing in all the sets in $\{\mathcal{U}_n:n\in[5]\}$. This property, together with the random permutations applied on each query set and the row-MDS property of $\bf H$, guarantees that any $T=2$ colluding servers can only observe that their query sets regarding each file have three common vectors and altogether $9$ linearly independent vectors in $\mathbb{F}_7^{10}$. Thus the roles of the desired and undesired files are disguised.

\subsubsection{The squeezing phase}
For $i\in [3]$, there are 2 redundant symbols among $\{Z_iW_{m^cn}: n\in [5]\}$  since the storage system is $(5,3)$-MDS. Let $W_{m^c}=({\bf x}_1, {\bf x}_2, {\bf x}_3)$. Due to the properties of the matrices ${\bf G}$, ${\bf H}$ and the query sets towards the undesired file, we have 
 \begin{align*}
 (U_1+U_4){\bf x}_1+(U_1+2U_4){\bf x}_2+(U_1+3U_4){\bf x}_3&=U_1({\bf x}_1+{\bf x}_2+{\bf x}_3)+U_4({\bf x}_1+2{\bf x}_2+3{\bf x}_3), \\
  (U_2+U_5){\bf x}_1+(U_2+2U_5){\bf x}_2+(U_2+3U_5){\bf x}_3&=U_2({\bf x}_1+{\bf x}_2+{\bf x}_3)+U_5({\bf x}_1+2{\bf x}_2+3{\bf x}_3), \\
  (U_3+U_6){\bf x}_1+(U_3+2U_6){\bf x}_2+(U_3+3U_6){\bf x}_3&=U_3({\bf x}_1+{\bf x}_2+{\bf x}_3)+U_6({\bf x}_1+2{\bf x}_2+3{\bf x}_3).
  \end{align*}
Thus, for $j\in [3]$, there is {\it at least} one redundant symbol among $\{\widetilde{U}_{n,j}W_{m^cn}: n\in [5]\}$. In total, there are at least $3\times2+3\times 1=9$ redundant symbols among the queried undesired symbols $\{\mathbb{\pi}_n'(\mathcal{U}_n)W_{m^cn}:n\in[5]\}$, and thus $\{\mathbb{\pi}_n'(\mathcal{U}_n)W_{m^cn}: n\in [5]\}$ span a linear space of dimension at most $I=21$.

When the $n^{\mathrm{th}}$ server responds to the query vectors, it does not directly respond with the queried symbols. Instead, it applies a combination strategy as follows. For $n\in[5]$, a preset deterministic full rank matrix $C_n$ of order $6$ is applied to transform
$$ (X_{n,1},\dots,X_{n,6})=\left(C_n\mathbb{\pi}_n(\mathcal{V}_n)W_{mn}\right)^\top,\quad (Y_{n,1},\dots,Y_{n,6}) = \left(C_n\mathbb{\pi}'_n(\mathcal{U}_n)W_{m^cn}\right)^\top. $$
Then a linear combination function $\mathcal{L}_n$ is used to compress  $(X_{n,1},\dots,X_{n,6},Y_{n,1},\dots,Y_{n,6})$ into $I_n$
{\it downloaded desired symbols} $\{X_{n,1},\dots,X_{n,{I_n}}\}$, $I_n$ {\it downloaded undesired symbols} $\{Y_{n,1},\dots,Y_{n,{I_n}}\}$, and $6-I_n$ {\it paired-up summation symbols} $\{X_{n,{I_n+1}}+Y_{n,{I_n+1}},\dots,X_{n,{6}}+Y_{n,{6}}\}$.

The key of the combination strategy is to guarantee that  $\sum_{n\in[5]}I_n=I=21$, and the 21 downloaded undesired symbols $\bigcup_{n\in [5]}{\{Y_{n,1},\dots,Y_{n,{I_n}}\}}$ should be able to linearly generate all the queried undesired symbols $\{\mathbb{\pi}_n'(\mathcal{U}_n)W_{m^cn}:n\in[5]\}$, regardless of the random permutations on the query sets over all servers. Once this is guaranteed, the user can eliminate all the interferences of the undesired file in the paired-up summations, and essentially gets all the queried desired symbols $\{\mathbb{\pi}_n(\mathcal{V}_n)W_{mn}:n\in[5]\}$. Since each $V_i$ is sent to exactly $K=3$ servers from a $(5,3)$-MDS system, the user could retrieve $V_iW_m$ for $i\in[10]$. Moreover, since $S$ is of full rank, retrieving $\{V_iW_m: i \in [10]\}$ is equivalent to retrieving the desired file $W_m$. In sum, the user retrieves the  desired file with PIR rate $R=\frac{30}{30+21}=\frac{10}{17}$, beating the FGHK-conjectured capacity $\frac{5}{9}$. 

It only remains to specifically design the matrices for the combination strategy. Set $I_1=I_2=I_3=6,~I_4=3,~I_5=0$, and
\begin{equation*}
    C_1=C_2=C_3=C_5=\left(\begin{array}{cccccc}
         1 & 0&0&0&0&0 \\
         0&1&0&0&0&0\\
         0&0&1&0&0&0\\
         0&0&0&1&0&0\\
         0&0&0&0&1&0\\
         0&0&0&0&0&1
    \end{array}\right),\quad
    C_4=\left(\begin{array}{cccccc}
         1 & 1&1&1&1&1 \\
         0&1&2&3&4&5\\
         0&1&4&2&2&4\\
         0&1&1&6&1&6\\
         0&1&2&4&4&2\\
         0&1&4&5&2&3
    \end{array}\right).
\end{equation*}

Each of the first three servers produces 6 downloaded desired symbols and 6 downloaded undesired symbols. The $4^{\mathrm{th}}$ server produces 3 downloaded desired symbols, 3 downloaded undesired symbols, and another 3 paired-up summation symbols. The $5^{\mathrm{th}}$ server produces 6 paired-up summation symbols. The role of the matrix $C_4$ is the most important. Here we pick $C_4$ to be an invertible Vandermonde matrix (and that is also why $\mathbb{F}_7$ suffices for this example). Regardless of the permutation applied on the query set for the $4^{\mathrm{th}}$ server, the 3 downloaded undesired symbols, together with the 18 downloaded  undesired symbols from the first three servers, will linearly generate all the queried undesired symbols $\{\mathbb{\pi}_n'(\mathcal{U}_n)W_{m^cn}:n\in[5]\}$. 

\subsection{Preparations}

In this subsection we present an auxiliary lemma, which plays a key role in designing the query vectors later.  

\begin{claim}\label{MDS}
Let ${\bf G}=[{\bf I}_{K} ~ {\bf A}]$ be a $K\times N$ matrix where ${\bf I}_K$ is the identity matrix of order $K$. Then {\bf G} has column-MDS property, i.e., any $K$ columns of {\bf G} are linearly independent, if and only if any subsquare matrix of ${\bf A}$ has full rank.
\end{claim}

This claim can be simply proved by elementary transformation of matrices (see~\cite{ramkumar2021codes}, Chapter 2, Section 4).

\begin{lemma}\label{lem:G-H}
Let $N\geq K+T$. Let ${\bf G} =[{\bf I}_{K} ~ {\bf g}_1 ~ {\bf g}_2 ~ \cdots ~ {\bf g}_{N-K}]$ be a $K\times N$ matrix with column-MDS property, where ${\bf g}_j=(g_{1,j},\dots,g_{K,j})^{\top} \in \mathbb{F}_q^{K\times 1}$ for $j\in[N-K]$. When $q$ is sufficiently large, there exists a matrix 
\begin{equation*}
{\bf H}=  
\left(\begin{array}{c}
{\bf{H}}'_{K\times T} \\
{\bf{I}}_{T} \\
{\bf{H}}''_{(N-K-T)\times T}
\end{array}
\right)_{N\times T}, 
\end{equation*}
with row-MDS property (any $T$ rows are linearly independent), where
\begin{equation*}
    {\bf H'}=\left(\begin{array}{cccc}
         \overline{h}_{1,1} & \overline{h}_{1,2}&\cdots&\overline{h}_{1,T} \\
        \overline{h}_{2,1} & \overline{h}_{2,2}&\cdots&\overline{h}_{2,T} \\
         \vdots & \vdots&\ddots&\vdots\\
         \overline{h}_{K,1}&\overline{h}_{K,2}&\cdots&\overline{h}_{K,T}
     \end{array}\right), ~ 
    {\bf H''}=\left(\begin{array}{cccc}
         h_{1,1} & h_{1,2}&\cdots& h_{1,T} \\
        h_{2,1} & h_{2,2}&\cdots&h_{2,T} \\
         \vdots & \vdots&\ddots&\vdots\\
        h_{N-K-T,1}&h_{N-K-T,2}&\cdots&h_{N-K-T,T}
     \end{array}\right),
\end{equation*}
and $\overline{h}_{k,t}=g_{k,t}+\sum_{\ell=1}^{N-K-T}g_{k,\ell+T}h_{\ell,t}$, for $k\in[K]$ and $t\in [T]$.
\end{lemma}
\begin{IEEEproof}
It suffices to find the appropriate variables $\{h_{\ell,t}: ~\ell\in[N-K-T],~t\in [T]\}$ such that ${\bf H}$ is row-MDS. For any $T$ rows of ${\bf H}$, consider the polynomial given by its determinant. Altogether there are ${N\choose T}$ polynomials and the variables are $\{h_{\ell,t}: ~\ell\in[N-K-T],~t\in [T]\}$. 

We claim that each polynomial corresponding to the determinant of some $T$ rows of ${\bf H}$ must be a nonzero polynomial. This can be done by choosing specific values for $\{h_{\ell,t}: ~\ell\in[N-K-T],~t\in [T]\}$ to make the $T$ rows an invertible matrix of order $T$. Suppose the $T$ rows contain $T_1$ rows of $\bf H'$, $T_2$ rows of $\bf I$, and $T_3$ rows of $\bf H''$. Then we can assign $\{h_{\ell,t}: ~\ell\in[N-K-T],~t\in [T]\}$ in a way such that the $T_2$ rows of $\bf I$ and the $T_3$ rows of $\bf H''$ together induce a permutation submatrix (with only one entry 1 in each row and at most one entry 1 in each column), and the rest entries in $\bf H''$ are all set to be 0. In this way, the invertibility relies on a subsquare matrix of $\bf H'$ of order $T_1$. This submatrix is in fact a submatrix of $[\bf g_1 \dots \bf g_{T}]$, and thus it is clearly invertible due to the MDS property of $\bf G$. 

Since all the ${N\choose T}$ polynomials are nonzero polynomials, then based on the celebrated Schwartz-Zippel lemma \cite{schwartz1980fast,zippel1979probabilistic}, there must exist a realization of variables $\{h_{\ell,t}: ~\ell\in[N-K-T],~t\in [T]\}$ over a sufficiently large finite field such that all the polynomials simultaneously evaluate to non-zero values. Thus, over a sufficiently large finite  field, for any systematic MDS matrix $\bf G$ one can find the corresponding desired row-MDS matrix ${\bf H}$.
\end{IEEEproof}

\begin{example}
Let $N=6,K=3,T=2$, and let
\begin{align*}
   {\bf G}= \left(\begin{array}{cccccc}
   1 &0&0&1&2&3  \\
   0&1&0&2&3&1\\
   0&0&1&3&2&1
    \end{array}\right)_{3\times 6}
\end{align*}
be a column-MDS matrix over $\mathbb{F}_{11}$. To find a desired matrix ${\bf H}_{6\times 2}$, the key is to choose the last row. The first three rows are then completely determined by $\mathbf{G}$ and the last row of $\mathbf{H}$. The choice of the last row of $\mathbf{H}$ must ensure that $\mathbf{H}$ is a row-MDS matrix. It is routine to check that by choosing the last row as $(1,1)$, the following matrix over $\mathbb{F}_{11}$ has the desired row-MDS property:
\begin{align*}
    {\bf H}=\left(\begin{array}{cc}
    4&5\\
    3&4\\
    4&3\\
    1&0\\
    0&1\\
    1&1     
    \end{array}\right).
\end{align*}   
\end{example}

Given an MDS storage code with a systematic generator matrix $\bf G$, the corresponding matrix $\bf H$ in Lemma \ref{lem:G-H} plays a key role in designing the queries, and later we will see why we need to set $\overline{h}_{k,t}$ in this way. In particular, when $K=N-T$ and ${\bf G}=[{\bf I}_{K} \;\; {\bf g}_1 \;\cdots\;{\bf g}_{T}]$ has MDS property, then according to \cref{MDS}, we can simple pick
\begin{align*}
  {\bf H}=  \left(\begin{array}{ccc}
{\bf g}_{1} & \cdots &{\bf g}_{T} \\
1 &\cdots & 0 \\
\vdots&\ddots&\vdots\\
0 & \cdots&1
\end{array}\right)_{N\times T}
\end{align*}
with row-MDS property.

\subsection{The general MDS-TPIR scheme for $(M,N,T,K)=(2,N,2,K)$, $N\geq K+2$}\label{subsec:main}

Now we are ready to present our general MDS-TPIR scheme for $(M,N,T,K)=(2,N,2,K)$, based on the disguise-and-squeeze approach. Since our framework relies on Lemma \ref{lem:G-H}, then we also need $N\geq K+2$.

Assume that each file consists of $L=\tbinom{N}{K}K=\tbinom{N-1}{K-1}N$ symbols from a sufficiently large finite field $\mathbb{F}_q$. Let 
\begin{align*}
W_1 = ({\bf a}_1, \dots, {\bf a}_{K}), \quad W_2 = ({\bf b}_1, \dots, {\bf b}_{K}),
\end{align*}
where ${\bf a}_1,\dots,{\bf a}_{K},{\bf b}_1,\dots,{\bf b}_{K} \in \mathbb{F}_q^{{N\choose K}\times 1}$ are vectors comprised of i.i.d. uniform symbols from $\mathbb{F}_q$. Each file is stored in the system via the same $(N,K)$-MDS code with generator matrix ${\bf G}_{K\times N}$, i.e.,
\begin{align*}
(W_{11},\dots, W_{1N})=W_1{\bf G}, \quad (W_{21},\dots, W_{2N})=W_2{\bf G},
\end{align*}
where $W_{mn}\in \mathbb{F}_q^{{N\choose K}\times1}$ for $m\in[2],n\in[N]$.

\subsubsection{The disguising phase (Constructions of queries)}
Let $W_m$ be the desired file and $W_{m^c}$ be the undesired file. To privately retrieve $W_m$, the user sends to the $n^{\mathrm{th}}$ server a set of queries $Q_n^{[m]}$, $n\in [N]$. The query $Q_n^{[m]}$ contains two parts $Q_n^{[m]}(W_m)$ and $Q_n^{[m]}(W_{m^c})$. Each part contains $\tbinom{N-1}{K-1}$ row vectors of length $\tbinom{N}{K}$, called query vectors. These query vectors are chosen according to the following steps. 
\begin{itemize}
\item The user privately chooses two matrices, $S$ and $S'$, independently and uniformly from the set of all full rank $\tbinom{N}{K}\times \tbinom{N}{K}$ matrices over $\mathbb{F}_q$. Let $\sigma\triangleq\tbinom{N-2}{K-2}$. Label the rows of $S$ as $\left\{V_i: i\in \left[\tbinom{N}{K}\right]\right\}$. Divide the rows of $S'$ into two parts and label them as $\left\{Z_i:i\in [\sigma]\right\}$ and $\left\{U_i:i\in\left[{N\choose K}-\sigma\right]\right\}$.
\item Define $\mathcal{V}_1,\mathcal{V}_2,\dots,\mathcal{V}_N$ to be the sets of query vectors regarding the desired file. Each $\mathcal{V}_n$ contains $\tbinom{N-1}{K-1}$ row vectors from $S$, and every $K$ sets out of $\{\mathcal{V}_1,\mathcal{V}_2,\dots,\mathcal{V}_N\}$ share exactly one common row vector from $S$. As a result, every two colluding servers share exactly $\sigma$ common query vectors,  and thereby altogether $2{N-1\choose K-1}-\sigma$ linearly independent vectors in $\mathbb{F}_q^{N\choose K}$ regarding the desired file. 

\item Next, define $\mathcal{U}_1,\mathcal{U}_2,\dots,\mathcal{U}_N$ to be the sets of query vectors regarding the undesired file. First, we put $\left\{Z_i:i\in [\sigma]\right\}$ into all the sets $\mathcal{U}_1,\mathcal{U}_2,\dots,\mathcal{U}_N$. Then let $\mu\triangleq {N-1\choose K-1} - \sigma$ and 
\begin{equation}\label{eqn:undesired}
\begin{small}
\left(
  \begin{array}{ccc}
    \widetilde{U}_{1,1} & \cdots & \widetilde{U}_{1,\mu} \\
    \widetilde{U}_{2,1} & \cdots & \widetilde{U}_{2,\mu} \\
    \vdots & \ddots & \vdots \\
    \widetilde{U}_{N,1} & \cdots &  \widetilde{U}_{N,\mu} \\
  \end{array}
\right) \triangleq {\bf H} \left(
  \begin{array}{ccc}
    U_{1} & \cdots & U_{\mu} \\
    U_{\mu+1} & \cdots & U_{2\mu} \\
  \end{array}
\right),
\end{small}
\end{equation}
and assign $\widetilde{U}_{n,1},\cdots,\widetilde{U}_{n,\mu}$ to $\mathcal{U}_n$, for $n\in[N]$. Here the matrix $\bf H$ is from \cref{lem:G-H} with the row-MDS property.  Note that Equation (\ref{eqn:undesired}) is well-defined since it is routine to check that ${N\choose K}-\sigma\geq 2\mu$. In this way, every two colluding servers share also exactly $\sigma$ common query vectors and altogether $2{N-1\choose K-1}-\sigma$ linearly independent vectors in $\mathbb{F}_q^{N\choose K}$ regarding the undesired file.

\item The query sent to the $n^{\mathrm{th}}$ server is given by 
\begin{equation*}
Q_n^{[m]} = (Q_n^{[m]}(W_m), Q_n^{[m]}(W_{m^c})) = (\mathbb{\pi}_n(\mathcal{V}_n),\mathbb{\pi}_n'(\mathcal{U}_n)),\end{equation*}
where $\mathbb{\pi}_n$ and $\mathbb{\pi}_n',~n\in [N]$ are random permutations independently and uniformly chosen in the symmetric group $S_{L/N}$ (recall that $L/N=\tbinom{N-1}{K-1}$). Write $\mathbb{\pi}_n(\mathcal{V}_n)$ and $\mathbb{\pi}_n'(\mathcal{U}_n)$ in the form of $(L/N)\times
\tbinom{N}{K}$ matrices, respectively.
\end{itemize} 

After receiving the queries, the $n^{\mathrm{th}}$ server projects its stored contents $(W_{mn}, W_{m^cn})$ along these query vectors, producing $L/N$ {\it queried desired symbols} $\mathbb{\pi}_n(\mathcal{V}_n)W_{mn}$ and $L/N$ {\it queried undesired symbols} $\mathbb{\pi}_n'(\mathcal{U}_n)W_{m^cn}$.

\subsubsection{The squeezing phase (Exploit the redundancy among the queried undesired symbols)}

Since each query vector $V_i$ is sent to exactly $K$ servers from an $(N,K)$-MDS system, the user could retrieve $V_iW_m$ from the queried desired symbols $\{\mathbb{\pi}_n(\mathcal{V}_n)W_{mn}:n\in[N]\}$. Since $S$ is of full rank, retrieving $\{V_iW_m:i\in [\tbinom{N}{K}]\}$ is equivalent to retrieving the file $W_m$ itself. Therefore, the scheme works as long as the user collects all the queried desired symbols $\{\mathbb{\pi}_n(\mathcal{V}_n)W_{mn}:n\in[N]\}$.

To maximize the PIR rate, each server does not directly respond with the queried symbols. Some queried symbols could be merged due to the redundancy among the queried undesired symbols. Clearly, for each $i\in[\sigma]$, there are $N-K$ redundant symbols among $\{Z_iW_{m^cn}:n\in[N]\}$ since the storage system is $(N,K)$-MDS. Let $W_{m^c}=({\bf x}_1, {\bf x}_2,\dots,{\bf x}_{K})$. For every $j$, $1\leq j\leq \mu$, due to the relations of $\bf G$ and $\bf H$ from \cref{lem:G-H} (in particular, the way we set $\overline{h}_{k,t}$ in $\bf H$), one could guarantee that 
\begin{align}\label{redundancy}
 \sum_{1\leq k\leq K} \widetilde{U}_{k, j} W_{m^ck} = \sum_{K+1\leq \ell \leq N} \widetilde{U}_{\ell, j}W_{m^c\ell},   
\end{align} 
which shows that there is {\it at least} one redundant symbol among $\{\widetilde{U}_{n, j}W_{m^cn}:n\in[N]\}$. Moreover, when $N>2K$, naturally there are $N-2K$ redundant symbols among $\{\widetilde{U}_{n, j}W_{m^cn}:n\in[N]\}$, since all these symbols can be expressed as linear combinations of $\{U_{j}{\bf x}_1,\dots,U_{j}{\bf x}_K,U_{\mu+j}{\bf x}_1,\dots,U_{\mu+j}{\bf x}_K\}$. Therefore, in total there are at least $\sigma(N-K)+\mu\zeta$ redundant symbols among the queried undesired symbols, where $\zeta=\max\{1,N-2K\}$. In other words, $\{\mathbb{\pi}_n'(\mathcal{U}_n)W_{m^cn}: n\in [N]\}$ span a linear space of dimension at most $I$, where $I\triangleq L-\sigma(N-K)-\mu\zeta$.

When the $n^{\mathrm{th}}$ server responds to the query vectors, it applies a combination strategy to map  $(\mathbb{\pi}_n(\mathcal{V}_n)W_{mn},\mathbb{\pi}_n'(\mathcal{U}_n)W_{m^cn})$ into an equal or smaller number of {\it downloaded symbols}. For $n\in [N]$, a preset deterministic full rank matrix $C_n$ of order $L/N$ is applied to transform 
$$ (X_{n,1},\dots,X_{n,{L/N}}) = (C_n\mathbb{\pi}_n(\mathcal{V}_n)W_{mn})^{\top}, ~ (Y_{n,1},\dots,Y_{n,{L/N}}) = (C_n\mathbb{\pi}_n'(\mathcal{U}_n)W_{m^cn})^{\top}.$$
Then a combination function $\mathcal{L}_n$ defined by
$$\mathcal{L}_n(X_{n,1},\dots,X_{n,{L/N}},Y_{n,1},\dots,Y_{n,{L/N}}) = (X_{n,1},\dots,X_{n,{I_n}},Y_{n,1},\dots,Y_{n,{I_n}},X_{n,{I_n+1}}+Y_{n,{I_n+1}},\dots,X_{n,{L/N}}+Y_{n,{L/N}})$$
will generate $I_n$ {\it downloaded desired symbols}, $I_n$ {\it downloaded undesired symbols}, and $L/N-I_n$ paired-up summation symbols. 

The key of the combination strategy is to guarantee that $\sum_{n\in[N]}I_n=I$, and the union of the downloaded undesired symbols should be able to linearly generate all the queried undesired symbols $\{\mathbb{\pi}_n'(\mathcal{U}_n)W_{m^cn}:n\in[N]\}$, regardless of the random permutations on the query sets over all servers. Once this property holds, the user is then able to eliminate the interferences of the undesired file in the paired-up summations and finally retrieves all the queried desired symbols $\{\mathbb{\pi}_n(\mathcal{V}_n)W_{mn}:n\in[N]\}$, and thereby the desired file itself. 

As mentioned in~\cite{sun2017private}, the choice of the matrices $C_n$'s to meet the constraints above is a highly non-trivial task, but is guaranteed to exist when the field is sufficiently large. We postpone the choices of the values $I_n$'s and the matrices $C_n$'s to the next subsection, in which we will apply a strategy different from~\cite{sun2017private} in order to reduce the field size.

Summing up the above, we have the following theorem. 

\begin{theorem}\label{th-1}
    For $(M,N,T,K)=(2,N,2,K)$, where $N\geq K+2$, there exists a PIR scheme with rate $R=\frac{N^2-N}{2N^2-2N+K^2-NK}$ when $N\leq 2K$ and $R=\frac{N^2-N}{N^2-N+2NK-K^2-K}$ when $N> 2K$.
\end{theorem}

\begin{IEEEproof}
We have explained the correctness of the scheme through the squeezing phase. As for the privacy, the constructions of the query sets $\mathcal{V}_1,\mathcal{V}_2,\dots,\mathcal{V}_N$ and $\mathcal{U}_1,\mathcal{U}_2,\dots,\mathcal{U}_N$, as well as the random permutations $\{\mathbb{\pi}_n,\mathbb{\pi}_n':n\in[N]\}$, together guarantee that for $m\in[2]$ and any two colluding servers $n_1$ and $n_2$, it holds that
$$(Q_{n_1}^{[1]}(W_m), Q_{n_2}^{[1]}(W_m))\sim (Q_{n_1}^{[2]}(W_m), Q_{n_2}^{[2]} (W_m)).$$
Here the random permutations ensure that the colluding servers cannot learn anything from the relative orders of the query vectors inside each query set. The rate of the scheme can be computed as $\frac{L}{L+I}$. Plugging in $L={N\choose K}K$, $\sigma={N-2\choose K-2}$, and $\mu={N-2\choose K-1}$, and the theorem follows. 
\end{IEEEproof}

\begin{remark}
Note that for $(M,N,T,K)=(2,N,2,K)$, the FGHK-conjectured capacity is $\frac{N}{N+K+1}$. A routine comparison shows that Theorem \ref{th-1} provides more counterexamples to the FGHK-conjecture when $2N-K-1>(N-K)^2$.    
\end{remark}

\begin{remark} 
When $N=K$, there is no redundancy in the system and the capacity is trivially $\frac{1}{2}$. For $(M,N,T,K)=(2,N,2,N-1)$, the protocol constructed by Sun and Jafar in \cite{sun2017private} achieves the capacity of $\frac{N^2-N}{2N^2-3N+2}$. Their scheme can also be seen as a special case of our scheme, but it should be noted that the queried undesired symbols corresponding to query vectors of type $\widetilde{U}$ provides no redundancy.
\end{remark}

\begin{remark}
    When the generator matrix ${\bf G'}$ of the MDS storage code is not in systematic form, we could transform ${\bf G'}$ into the systematic form ${\bf G}$, and then use \cref{lem:G-H} to obtain the matrix ${\bf H}$ corresponding to ${\bf G}$. The matrix ${\bf H}$ could be directly used as described in the protocol above to construct the query vectors for the MDS storage system with generator matrix ${\bf G'}$, since Equation \eqref{redundancy} still holds and thus the squeezing phase still works. Therefore, our PIR scheme holds for any generator matrix, not necessarily in systematic form. 
\end{remark}  

\begin{remark}
In~\cite{holzbaur2021toward} it was proved that the FGHK conjecture is correct when restricting to PIR schemes which are linear and full support-rank. The term ``full support-rank" requires that for any set of $T$ servers, the query vectors regarding each file are linearly independent. The fact that our scheme (and also the results in \cite{sun2017private}) does not satisfy the full support-rank condition leads to breaking the barrier of the FGHK conjecture. 
\end{remark}

\subsection{The selection of $C_n$'s over a small field}\label{subsec:field}
In the work by Sun and Jafar~\cite{sun2017private}, for the case $(M,N,T,K)=(2,4,2,2)$, they set $I_n=2$ for all $n\in[4]$ and thus each server responds with 5 downloaded symbols: 2 desired, 2 undesired, and one paired-up summation. The matrices $\{C_n:n\in[4]\}$ are full-rank and should guarantee that the 8 downloaded undesired symbols could linearly generate all the queried undesired symbols, regardless of the how the three queries in each query set $\mathcal{U}_n$ are permuted. Since there are altogether $3!^4=1296$ possibilities, such matrices are difficult to find and finally Sun and Jafar only managed to find the desired matrices over a relatively large field $\mathbb{F}_{349}$. 

We observe that the field size can be dramatically reduced if we do not require the download cost to be uniform among all the servers. Continuing with the notations in Example~\ref{emp:SunJafar} and our general framework, let $
I_1 = 3, ~ I_2 = 3,  ~I_3 = 2, ~ I_4 = 0,$ and 
\begin{align*}
    C_1=C_2=C_4=\left(\begin{array}{ccc}
         1 & 0&0 \\
         0&1&0\\
         0&0&1
    \end{array}\right),\quad
    C_3=\left(\begin{array}{ccc}
        1 & 1&1 \\
         0&1&2\\
         0&1&1
    \end{array}\right).
\end{align*} 
In this way, the scheme can be defined simply over  $\mathbb{F}_3$. The user essentially collects 8 downloaded undesired symbols: 
$$\{Z_1\mathbf{c},(U_1+U_2)\mathbf{c},(U_3+U_4)\mathbf{c},Z_1\mathbf{d},(U_1+2U_2)\mathbf{d},(U_3+2U_4)\mathbf{d},F_1(\mathbf{c+d}),F_2(\mathbf{c+d})\},$$
where $F_1,F_2$ are the first two vectors of $C_3\pi_3'(\mathcal{U}_3)$. It is routine to check that for any $\pi_3'\in S_3$, these 8 downloaded undesired symbols could linearly generate all the queried undesired symbols $\{\mathbb{\pi}_n'(\mathcal{U}_n)W_{2n}:n\in[4]\}$ as expected. The interferences of the undesired file in the 4 summation symbols (1 from Server III and 3 from Server IV) can be eliminated, and thus the user successfully collects all the queried desired symbols.

One may argue that the nonuniform download cost among servers is a blemish in our scheme. However, we have at least three reasons to keep the non-uniformity. First, it should be noted that for certain parameters the value $I$ is not necessarily  divisible by $N$ (for example, in Subsection \ref{(2,5,2,3)}, for $(M,N,T,K)=(2,5,2,3)$ we have $I=21$ which is not  divisible by $N=5$) and thus it is natural to have nonuniform download cost, especially when we want to minimize the sub-packetization level of the scheme. Second, if one insists on constructing a PIR scheme with uniform download cost among servers, one could just repeat our scheme $N$ times (and thus the sub-packetization level is also $N$ times larger) and cyclically permute the servers in terms of their download cost, which would finally lead to a PIR scheme with uniformity. Third, and most importantly, setting $I_n$'s in a uniform way would require each server to contribute some downloaded undesired symbols (essentially, a linear subspace of the linear space spanned by the queried undesired symbols on that server), and thus we need to coordinate all the matrices $C_n$'s to guarantee that the downloaded undesired symbols from all servers linearly generate all the queried undesired symbols, regardless of the permutations on each query set. A total of ${(L/N)!}^N$ realizations of permutations should be considered, and this exponentially explosive number would make a brute-force search for the matrices computationally infeasible. In fact, as we will demonstrate now, the key to reduce the size of the field is a result of the nonuniform download cost. Now we explain how to choose the values of $I_n$'s and the matrices $C_n$'s in general.

Recall the query sets regarding the undesired file in our scheme. For the queried undesired symbols $\{Z_iW_{m^cn}:i\in[\sigma],n\in[N]\}$, we can simply download only from the first $K$ servers and consider those on the remaining $N-K$ servers as redundancy. For each $j\in[\mu]$, the queried undesired symbols $\{\widetilde{U}_{n,j} W_{m^c n} : n \in [N]\}$ contain at least $\zeta=\max\{1,N-2K\}$ redundant symbols. Without loss of generality, by reordering the servers, we consider the last $\zeta$ servers as redundancy. Then we can download only from the first $N-\zeta=\min\{N-1,2K\}>K$ servers. As a result, we set the values of $I_n$'s as follows:
$$I_n=\begin{cases}
    {N-1\choose K-1}, & \mbox{if } 1\leq n \leq K, \\
    {N-2\choose K-1}, & \mbox{if } K < n \leq N-\zeta, \\
    0, & \mbox{if } N-\zeta <  n \leq N.
        \end{cases}
$$
Each of the first $K$ servers produces ${N-1\choose K-1}$ downloaded desired symbols and ${N-1\choose K-1}$ downloaded undesired symbols. Each of the last $\zeta$ servers produces ${N-1\choose K-1}$ pairwise summation symbols. For these servers, we can simply choose $C_n$ as the identity matrix. As for the intermediate $N-\zeta-K$ servers, each server produces ${N-2\choose K-1}$ downloaded desired symbols, ${N-2\choose K-1}$ downloaded undesired symbols, as well as ${N-2\choose K-2}$ pairwise summation symbols. We need to design the matrices $C_n$'s, for $K<n\leq N-\zeta$, such that the downloaded undesired symbols from all servers could linearly generate all the queried undesired symbols $\{\mathbb{\pi}_n'(\mathcal{U}_n)W_{m^c n}: n\in[N]\}$, regardless of the permutations used on all the query sets. This can be done by choosing the matrices $C_n$'s, for $K<n\leq N-\zeta$, as invertible Vandermonde matrices of order ${N-1\choose K-1}$. Formally, we have the following theorem.
\begin{theorem}
The $I_n$'s and $C_n$'s defined above satisfy the following conditions:
\begin{enumerate}
    \item[{1.}] For every $n\in[N]$, the matrix $C_n$ is full-rank.
    \item[{2.}] $\sum_{n\in[N]}I_n=I$.
    \item[{3.}] The union of the downloaded undesired symbols can linearly generate all the queried undesired symbols 
   $\{\mathbb{\pi}_n'(\mathcal{U}_n) W_{m^c n} : n \in [N]\}$, regardless of the random permutations on the query sets over all servers. 
\end{enumerate}
\end{theorem}

\begin{IEEEproof}
The first item is obvious. Since $I=L-\sigma(N-K)-\mu\cdot \zeta=N{N-1\choose K-1}-{N-2\choose K-2}(N-K)-{N-2\choose K-1}\zeta=K{N-2\choose K-2}+(N-\zeta){N-2\choose K-1}=K{N-1\choose K-1}+(N-\zeta-K){N-2\choose K-1}=\sum_{n\in[N]}I_n$, the second item follows. 

Regardless of how the query vectors are permuted, the symbols $\{Z_i W_{m^c n} : i\in [\sigma],~n \in [K]\}$ and $\{\widetilde{U}_{n,j} W_{m^c n} : j\in [\mu],~n \in [K]\}$ can be directly downloaded from the first $K$ servers. Clearly, due to the MDS property of the storage system, $\{Z_i W_{m^c n} : i\in [\sigma],~n \in [K]\}$ could linearly generate all the symbols in $\{Z_i W_{m^c n} : i\in [\sigma],~n \in [N]\}$. For the $n^{\mathrm{th}}$ server, $n \in \{K+1, \dots, N-\zeta\}$, the first $\mu={N-2\choose K-1}$ symbols of $C_n  \mathbb{\pi}_n'(\mathcal{U}_n) W_{m^c n}$ are directly downloaded, where $C_n$ is an invertible Vandermonde matrix. After eliminating the interferences of the symbols $\{Z_iW_{m^cn}:i\in[\sigma]\}$ from the $\mu$ downloaded undesired symbols on the $n^{\mathrm{th}}$ server, we could retrieve $\{\widetilde{U}_{n,j} W_{m^c n} : j\in [\mu]\}$. Thus, the user has downloaded $\{\widetilde{U}_{n,j} W_{m^c n} : j\in [\mu],~n \in [N-\zeta]\}$, which will then linearly generate $\{\widetilde{U}_{n,j} W_{m^c n} : j\in [\mu],~n \in [N]\}$. To sum up, the union of the downloaded undesired symbols can linearly generate all the queried undesired symbols $\{\mathbb{\pi}_n'(\mathcal{U}_n) W_{m^c n} : n \in [N]\}$. 
\end{IEEEproof}

Therefore, our PIR scheme could work on the finite field $\mathbb{F}_q$ where $q$ is the minimum prime power larger than or equal to ${N-1\choose K-1}$. Whether or not the field size can be further reduced is left for future research.

\section{There is still more to squeeze}\label{Sec:squeeze}

In Subsection \ref{subsec:main}, when we exploit the redundancy, we only claim that there are {\it at least} $\sigma(N-K)+\mu\cdot\max\{1,N-2K\}$ redundant symbols among the queried undesired symbols. Apparently, the more redundancy we can squeeze from the queried undesired symbols, the better PIR rate we could derive. The exact number of redundant symbols among the queried undesired symbols might depend on the storage code and the specific queries, and seems to be a rather non-trivial problem. In this section, we consider a special case when the storage code is a generalized Reed-Solomon (GRS) code, and propose a method to squeeze maximum redundancy. Note that restricting the MDS storage system to be GRS is already a common practice in the study of PIR, say for example, in the paper~\cite{freij2017private} proposing the FGHK conjecture.  
 
\begin{definition}[GRS codes]
Given a finite field $\mathbb{F}_q$, for integers $N$ and $K$ such that $1 \leq K \leq N \leq q$, choose $N$ distinct elements $\alpha_1, \alpha_2, \dots, \alpha_N \in \mathbb{F}_q$ and $N$ nonzero (not necessarily distinct) elements $v_1, v_2, \dots, v_N \in \mathbb{F}_q^*$. 
The generalized Reed-Solomon code $\mathrm{GRS}_{K}(\boldsymbol{\alpha}, \bf{v})$ is the set of codewords
\[
(v_1 f(\alpha_1), v_2 f(\alpha_2), \dots, v_N f(\alpha_N)),
\]
where $f(x)$ ranges over all polynomials of degree less than $K$ over $\mathbb{F}_q$. The minimum distance of this code is $d = N - K + 1$. 
\end{definition}

\begin{definition}[Schur Products]
Let $C_1, C_2 \subseteq \mathbb{F}_q^N$ be two linear codes. The \emph{Schur product} of $C_1$ and $C_2$, denoted $C_1 \star C_2$, is the linear code defined as the span of the component-wise (element-wise) products of codewords from $C_1$ and $C_2$, i.e.,
\[
C_1 \star C_2 = \mathrm{span}\{ \mathbf{c}_1 \star \mathbf{c}_2 | \mathbf{c}_1 \in C_1, \mathbf{c}_2 \in C_2 \},
\]
where $\mathbf{c}_1 \star \mathbf{c}_2$ denotes the component-wise product of codewords $\mathbf{c}_1$ and $\mathbf{c}_2$. 
\end{definition}

Let $\langle \mathbf{x}, \mathbf{y}\rangle$ be the inner product of two vectors. The next lemma is the key to the design of queries for GRS-coded systems, aiming for the maximum redundancy among queried undesired symbols. 

\begin{lemma}\label{GRS}
Let ${\bf G}_{K\times N}$ and ${\bf G'}_{T\times N}$ be the generator matrices of codes $C$ and $C'$, respectively, where $C=\mathrm{GRS}_{K}(\boldsymbol{\alpha}, \bf{v})$ and $C'=\mathrm{GRS}_{T}(\boldsymbol{\alpha}, \bf{w})$. For any two matrices $\mathbf{W}\in \mathbb{F}_q^{\ell \times K}$ and $\mathbf{U}\in \mathbb{F}_q^{\ell \times T}$, define 
    \begin{align*}
        (\mathbf{x}_1,\dots,\mathbf{x}_N)=\mathbf{W}{\bf G}=\begin{pmatrix}
            x_{1,1}& \cdots & x_{1,N}\\
            \vdots&\ddots&\vdots \\
            x_{\ell,1}&\cdots&x_{\ell,N}
        \end{pmatrix},  ~ ~ (\mathbf{y}_1,\dots,\mathbf{y}_N)=\mathbf{U}{\bf G'}=\begin{pmatrix}
            y_{1,1}&\cdots&y_{1,N}\\
            \vdots&\ddots&\vdots\\
            y_{\ell,1}&\cdots&y_{\ell,N}
        \end{pmatrix}.
    \end{align*}

Then the set $\{\langle {\mathbf x}_1,{\mathbf y}_1\rangle,\dots,\langle {\mathbf x}_N,{\mathbf y}_N\rangle\}$ contains exactly $\max\{N-K-T+1,0\}$ redundant symbols. 
\end{lemma}

\begin{IEEEproof}
    For $n\in[N]$, $\langle {\mathbf x}_n, {\mathbf y}_n\rangle=\sum_{i=1}^{\ell}x_{i,n}y_{i,n}$. We have 
    \begin{align*}
        \left(\langle {\mathbf x}_1,{\mathbf y}_1\rangle,\dots,\langle {\mathbf x}_N,{\mathbf y}_N\rangle\right)=\sum_{i=1}^{\ell}\left(x_{i,1}y_{i,1},\dots,x_{i,N}y_{i,N}\right) \in C\star C'. 
    \end{align*}
    By \cite[Proposition 3]{freij2017private}, we have $\mathrm{GRS}_{K}({\boldsymbol{\alpha}}, {\bf{v}})\star \mathrm{GRS}_{T}({\boldsymbol{\alpha}}, {\bf{w}})=\mathrm{GRS}_{\min\{K+T-1,N\}}(\boldsymbol{\alpha}, {\bf{v}}\star{\bf{w}})$ and the lemma follows.
\end{IEEEproof}

Look back on our PIR scheme in Section \ref{Sec:main} for $(M,N,T,K)=(2,N,2,K)$ where $N\geq K+2$.
If the storage code is $\mathrm{GRS}_{K}(\boldsymbol{\alpha}, \bf{v})$, we can take the transpose of ${\bf H}$ as the generator matrix of $\mathrm{GRS}_{2}(\boldsymbol{\alpha}, \bf{w})$. Then there are $N-K-1$ redundant symbols in the set $\{\widetilde{U}_{n, j}W_{m^cn}:n\in[N]\}$, for every $j\in[\mu]$. Hence, in total there are $\sigma(N-K)+\mu(N-K-1)$ redundant symbols among the queried undesired symbols $\{\mathbb{\pi}_n'(\mathcal{U}_n)W_{m^cn}:n\in[N]\}$. In other words, the queried undesired symbols span a linear space of dimension  $I\triangleq L-\sigma(N-K)-\mu(N-K-1)$. By substituting $L={N\choose K}K$, $\sigma={N-2\choose K-2}$, and $\mu={N-2\choose K-1}$, into the rate $R=\frac{L}{L+I}$, we obtain the following theorem.

\begin{theorem}\label{thm:GRS}
For $(M,N,T,K)=(2,N,2,K)$, $N\geq K+2$, if the storage code is a GRS code, then there exists a PIR scheme with rate $R=\frac{N^2-N}{N^2+KN-2K}$.
\end{theorem}

\begin{remark}
Note that after squeezing more redundancy, the rate above always exceeds the FGHK-conjectured capacity value $\frac{N}{N+K+1}$. Therefore, for $(M,N,T,K)=(2,N,2,K)$, $N\geq K+2$, when the storage code is a GRS code, our scheme provides an infinite class of counterexamples for the FGHK conjecture and also beats the state-of-the-art result in~\cite{d2019one} and~\cite{holzbaur2021toward}.    
\end{remark}

\begin{remark}
In particular, when further restricting to $K=2$, an easier way is as follows. Suppose the storage system is based on an arbitrary MDS matrix $\bf{G}$ (let the two row vectors be $\bf{g}_1$ and $\bf{g}_2$), then we can simply set $\bf{H}=\bf{G}^{\top}$. The Schur product of the two linear codes determined by $\bf{G}$ and $\bf{H}^{\top}$ can be spanned by $\{\bf{g}_1\star\bf{g}_1,\bf{g}_1\star\bf{g}_2,\bf{g}_2\star\bf{g}_2\}$ and is thus of dimension 3. Thus there are $N-3$ redundant symbols in the set $\{\widetilde{U}_{n, j}W_{m^cn}:n\in[N]\}$, for every $j\in[\mu]$. Hence we derive the same result as Theorem \ref{thm:GRS} for $K=2$.
\end{remark}

We present a concrete example for $(M,N,T,K)=(2,5,2,2)$. For this parameter our general scheme only has rate $10/17$, less than the FGHK-conjectured value $5/8$. When the storage code is GRS, by squeezing more we can beat the FGHK-conjectured value.

\begin{example}\label{ex(5,2,2)}
This example could work on the finite field $\mathbb{F}_5$. Each file is comprised of $L=20$ independent symbols from $\mathbb{F}_5$:
\begin{align*}
    W_1=({\bf a}_1, {\bf a}_2)\quad W_2=({\bf b}_1, {\bf b}_2),
\end{align*}
where ${\bf a}_1, {\bf a}_2,{\bf b}_1, {\bf b}_2\in \mathbb{F}_5^{10\times1}$. Each file is stored in the system via a $(5,2)$-MDS code. Let 
\begin{align*}
    {\bf G}=\left(\begin{array}{ccccc}
        1 &1&1&1&1  \\
         0&1&2&3&4 
    \end{array}\right),\; {\bf H}=\left(\begin{array}{cc}
      1   &0  \\
       1  &1\\
       1&2\\
       1&3\\
       1&4
    \end{array}\right).
\end{align*}
Both ${\bf G}$ and ${\bf H}^{\top}$ are generator matrices of GRS codes with the same evaluation points.
The storage is specified as
\begin{equation*}
 (W_{11},\dots, W_{15})=W_1{\bf G}, \quad (W_{21},\dots, W_{25})=W_2{\bf G}.   
\end{equation*}
where $W_{mn}\in \mathbb{F}_5^{10\times1}$ for $m\in[2],n\in[5]$.

\subsubsection{The disguising phase} 
Let $S$ and $S'$ be independently and uniformly chosen from the set of all full rank $10\times 10$ matrices over $\mathbb{F}_5$. Label the rows of $S$ as $\{V_i: i\in [10]\}$ and the rows of $S'$ as $\{Z_1, U_1, \dots, U_9\}$. Define
\begin{align*}
\mathcal{V}_1&=\{V_1, V_2, V_3, V_4\},\quad ~ \mathcal{U}_1=\{Z_1, \widetilde{U}_{1,1}, \widetilde{U}_{1,2},\widetilde{U}_{1,3}\},\\
\mathcal{V}_2&=\{V_1, V_5, V_6,V_7\},\quad ~\mathcal{U}_2=\{Z_1, \widetilde{U}_{2,1}, \widetilde{U}_{2,2},\widetilde{U}_{2,3}\},\\
\mathcal{V}_3&=\{V_2, V_5, V_8,V_9\},\quad ~\mathcal{U}_3=\{Z_1, \widetilde{U}_{3,1},\widetilde{U}_{3,2},\widetilde{U}_{3,3}\},\\
 \mathcal{V}_4&=\{V_3, V_6, V_8,V_{10}\},\quad\mathcal{U}_4=\{Z_1, \widetilde{U}_{4,1},\widetilde{U}_{4,2},\widetilde{U}_{4,3}\},\\
\mathcal{V}_5&=\{V_4, V_7, V_9,V_{10}\},\quad\mathcal{U}_5=\{Z_1, \widetilde{U}_{5,1},\widetilde{U}_{5,2},\widetilde{U}_{5,3}\},
\end{align*}
where 
   \begin{eqnarray*}
\begin{array}{l}
\left(\begin{array}{ccc}\widetilde{U}_{1,1} & \widetilde{U}_{1,2} & \widetilde{U}_{1, 3} \\ 
 
\vdots &\vdots & \vdots\\
\widetilde{U}_{5, 1} & \widetilde{U}_{5,2} & \widetilde{U}_{5, 3}\end{array}\right) 
\triangleq{\bf H}\left(\begin{array}{lll}U_{1} & U_2 & U_{3} \\
U_{4} & U_5 & U_{6}\end{array}\right).
\end{array}
\end{eqnarray*} 

Let $W_m$ be the desired file and  $W_{m^c}$ be  the undesired file. For $\forall n\in[5]$, $Q_n^{[m]} = (Q_n^{[m]}(W_m), Q_n^{[m]}(W_{m^c}))=(\mathbb{\pi}_n(\mathcal{V}_n),\mathbb{\pi}_n'(\mathcal{U}_n))$, where $\mathbb{\pi}_n$ and $\mathbb{\pi}_n'$ are independently and uniformly random permutations in $S_4$. After receiving the queries, the $n^{\mathrm{th}}$ server applies the query vectors to its stored contents $(W_{mn},\, W_{m^c n})$, producing $4$ queried desired symbols $\mathbb{\pi}_n(\mathcal{V}_n)W_{mn}$ and 
$4$ queried undesired symbols $\mathbb{\pi}_n'(\mathcal{U}_n)W_{m^cn}$.

%Now we briefly analyze how the queries keep the privacy by disguising. Here $\{\mathcal{V}_n:n\in[5]\}$ are the query sets for the desired file, where every two sets share exactly one common vector from $S$. Consequently, every two sets in $\{\mathcal{V}_n:n\in[5]\}$ share exactly one common vector, and this key property is imitated in $\{\mathcal{U}_n:n\in[5]\}$, the query sets for the undesired file, by the vector $Z_1$ appearing in all the sets in $\{\mathcal{U}_n:n\in[5]\}$. This property, together with the random permutations applied on each query set and the row-MDS property of $\bf H$, guarantees that any $T=2$ colluding servers can only observe that their query sets regarding either file have one common vector and altogether $7$ linear independent vectors in $\mathbb{F}_5^{10}$. Thus the roles of the desired and undesired files are disguised.  

\subsubsection{The squeezing phase}
There are $3$ redundant symbols among $\{Z_1W_{m^cn}:n\in[5]\}$ since the storage code is $(5,2)$-MDS. As for the other queried undesired symbols, the design of $\bf H$ helps us squeeze more. Let $W_{m^c}=({\bf x}_1,{\bf x}_2)$.
For $\forall j\in [3]$,
\begin{align*}
   \left( \begin{array}{l}
       \widetilde{U}_{1,j}{\bf x}_1   \\
          \widetilde{U}_{2,j}({\bf x}_1+{\bf x}_2) \\
           \widetilde{U}_{3,j}({\bf x}_1+2{\bf x}_2 )\\
        \widetilde{U}_{4,j}({\bf x}_1+3{\bf x}_2 )\\ 
         \widetilde{U}_{5,j}({\bf x}_1+4{\bf x}_2 ) 
    \end{array}\right)=\left(\begin{array}{cccc}
    1&0&0&0\\
    1&1&1&1\\
    1&2&2&4\\
    1&3&3&4\\
    1&4&4&1
    \end{array}\right)\left(
    \begin{array}{l}
    U_{j}{\bf x}_1\\
    U_{j+3}{\bf x}_1\\
    U_{j}{\bf x}_2\\
    U_{j+3}{\bf x}_2
    \end{array}\right).
\end{align*}
Note that the rank of the coefficient matrix in the equation above is only $3$. Thus, there are two redundant symbols among $\{ \widetilde{U}_{1,j}{\bf x}_1 ,\widetilde{U}_{2,j}({\bf x}_1+{\bf x}_2),\widetilde{U}_{3,j}({\bf x}_1+2{\bf x}_2 ),\widetilde{U}_{4,j}({\bf x}_1+3{\bf x}_2 ),\widetilde{U}_{5,j}({\bf x}_1+4{\bf x}_2 )\}$ for each $j\in[3]$. Hence, altogether there are $3+3\times 2=9$ redundant symbols among the queried undesired symbols. In other words, $\{\mathbb{\pi}_n'(\mathcal{U}_n)W_{m^cn}:n\in[5]\}$ span a linear space of dimension at most $I=11$.

The combination strategy is similar as before.
Specifically, we take $I_1=4,~I_2=4,~I_3=3,~I_4=0,~I_5=0$, and
\begin{align*}
    C_1=C_2=C_4=C_5=\left(\begin{array}{cccc}
         1 & 0&0&0 \\
         0&1&0&0\\
         0&0&1&0\\
         0&0&0&1
    \end{array}\right),\quad 
    C_3=\left(\begin{array}{cccc}
         1 & 1&1&1 \\
         0&1&2&3\\
         0&1&4&4\\
         0&1&3&2
    \end{array}\right).
\end{align*} 
Here $C_3$ is an invertible Vandermonde matrix, guaranteeing that the 11 downloaded undesired symbols could linearly generate all the queried undesired symbols, regardless of the permutation on the query set for the $3^{\mathrm{rd}}$ server. The rate of the scheme is $20/31$, beating the value  $10/17$ of our general result in Theorem \ref{th-1} and also the FGHK-conjectured value $5/8$.
\end{example}

%{\color{red} TBD} In fact, in Section xxx we will show that $20/31$ is the exact linear PIR capacity for $(M,N,T,K)=(2,5,2,2)$.

\section{Some generalizations of our scheme}\label{Sec:general}

In this section, we consider several possible generalizations of our disguise-and-squeeze MDS-TPIR scheme.

\subsection{$M$ could be arbitrary}

The generalization from two files to an arbitrary number of files is rather straightforward. In the disguising phase, $M$ families of query sets (one family for each file) are generated independently, where the sets have the same setup like $\mathcal{V}$ for the desired file and the sets have the same setup like $\mathcal{U}$ for all the $M-1$ undesired files. For each server, $M$ permutations have to be generated independently, each one for permuting the queries for each file. In the squeezing phase, the only difference is that in the combination strategy we have $M$-wise summations instead of pairwise summations. The interferences in the $M$-wise summations could be eliminated based on the downloaded undesired symbols for all the $M-1$ undesired files. In this way, for arbitrary $M$, Theorems \ref{th-1} and \ref{thm:GRS} can be extended as follows.

\begin{theorem}\label{thm:M}
For $(M,N,T=2,K)$ where $N\geq K+2$, there exists a PIR scheme with rate $R=\frac{N^2-N}{N^2-N+(M-1)(N^2-N+K^2-NK)}$ when $N\leq 2K$ and $R=\frac{N^2-N}{N^2-N+(M-1)(2NK-K^2-K)}$ when $N> 2K$. Moreover, if the storage code is a GRS code, then there exists a PIR scheme with rate $R=\frac{N^2-N}{N^2-N+(M-1)(NK+N-2K)}$.
\end{theorem}

\begin{IEEEproof}
The correctness and the privacy of the scheme are exactly the same as before. The rate of the scheme has changed from $\frac{L}{L+I}$ to $\frac{L}{L+(M-1)I}$. Here $I=L-\sigma(N-K)-\mu\zeta$ represents the number of downloaded undesired symbols for each of the $M-1$ undesired files. Similar as the previous analysis, for arbitrary MDS-coded system we have $\zeta=1$ when $N\leq 2K$ and $\zeta=N-2K$ when $N> 2K$, and for GRS-coded system we have $\zeta=N-K-1$. Plugging in $L={N\choose K}K$, $\sigma={N-2\choose K-2}$, and $\mu={N-2\choose K-1}$, and the theorem follows.   
\end{IEEEproof}

Unfortunately, we have to admit that the current scheme for arbitrary $M$ does not deserve much attention, since the rate is usually not satisfactory. However, if we consider the $P$-out-of-$M$ multi-file retrieval, then we can make some progress.

\subsection{The multi-file version}\label{subsec:multi}

The multi-file generalization also involves setting up $M$ families of query sets (one family for each file). Now the $P$ families of sets regarding the desired files have the same setup like $\mathcal{V}$, and the other $M-P$ families of sets regarding the undesired files have the same setup like $\mathcal{U}$. On the $n^{\mathrm{th}}$ server the permutations $\pi_{1,n},\dots,\pi_{M,n}$ are generated independently, where $\pi_{m,n}$ is used to permute the queries regarding the $m^{\mathrm{th}}$ file. The preset full rank matrix $C_n$'s are the same as in the single-file retrieval scheme, and they are used to transform the queried symbols of each file, either the desired $\pi_{m,n}(\mathcal{V}_n)W_{mn}$ or the undesired $\pi_{m,n}(\mathcal{U}_n)W_{mn}$, into $L/N$ symbols $X_{m,n,1},\dots,X_{m,n,{L/N}}$.

The vital difference lies in the combination function $\mathcal{L}'_n$ for each $n\in[N]$. For the $P$-out-of-$M$ multi-file retrieval, the combination function $\mathcal{L}'_n$ relies on a fixed generator matrix ${\bf D}_{P\times M}$ of an $(M, P)$-MDS code. Out of the $ML/N$ symbols $\{X_{m,n,i}:m\in[M],i\in[L/N]\}$, only the $MI_n$ symbols $\{X_{m,n,i}:m\in[M],i\in[I_n]\}$ are directly downloaded. The rest are combined into $P(L/N-I_n)$ pieces of $M$-wise summations defined by
$\left\{ \mathbf{D} (X_{1,n,i+I_n}, X_{2,n,i+I_n},\dots,X_{M,n,i+I_n})^{\top}: i\in\left[L/N-I_n\right]\right\}$. Once the interferences of the undesired files in these $M$-wise summations are eliminated based on the directly downloaded undesired symbols, we have $L/N-I_n$ systems of linear equations. Each system has $P$ unknown variables (corresponding to desired files) and $P$ equations. The MDS property of the matrix $\bf D$ allows the user to solve all the remaining queried desired symbols. The total number of downloaded symbols is $\sum_{n\in[N]} \left(P(L/N-I_n)+MI_n\right) = PL+(M-P)I$, and thus the rate of our multi-file scheme is $R=\frac{PL}{PL+(M-P)I}$. Plugging in $L=\tbinom{N}{K}K$ and $I=L-\sigma(N-K)-\mu\zeta$. Similar as the previous analysis, for arbitrary MDS-coded system we have $\zeta=1$ when $N\leq2K$ and $\zeta=N-2K$ when $N> 2K$, and for GRS-coded system we have $\zeta=N-K-1$. Then we have the following result.

\begin{theorem}
For $(M,N,T=2,K)$ where $N\geq K+2$, there exists a $P$-out-of-$M$ multi-file PIR scheme with rate $R=\frac{P(N^2-N)}{P(N^2-N)+(M-P)(N^2-N+K^2-NK)}$ when $N\leq 2K$ and $R=\frac{P(N^2-N)}{P(N^2-N)+(M-P)(2NK-K^2-K)}$ when $N> 2K$. Moreover, if the storage code is a GRS code, then there exists a $P$-out-of-$M$ multi-file PIR scheme with rate $R=\frac{P(N^2-N)}{P(N^2-N)+(M-P)(NK+N-2K)}$. 
\end{theorem}

\begin{remark}
Now we compare our result with the multi-file PIR scheme of rate~$\frac{P\tbinom{N}{K}}{M\tbinom{N}{K}-(M-P)\tbinom{N-2}{K}}$ in~\cite{zhang2017private}. When the storage system uses an arbitrary MDS code, it is routine to check the two rates are the same when $N>2K$ and our new result is slightly better when $N\leq 2K$. When the storage code is GRS, then our new result always beats the scheme in~\cite{zhang2017private}.  
\end{remark}

Now we present an example for $(P,M,N,T,K)=(2,3,5,2,2)$ and for simplicity set the storage code to be GRS. Assume the user wants to retrieve both $W_1$ and $W_2$.
\begin{example}
 Each file is comprised of $L=20$ independent symbols from $\mathbb{F}_5$:
\begin{align*}
    W_1=({\bf a}_1, {\bf a}_2)\quad W_2=({\bf b}_1, {\bf b}_2),\quad  W_3=({\bf c}_1, {\bf c}_2),
\end{align*}
where ${\bf a}_1, {\bf a}_2,{\bf b}_1, {\bf b}_2, {\bf c}_1, {\bf c}_2\in \mathbb{F}_5^{10\times1}$. Each file is stored in the system via a $(5,2)$-MDS code. Let 
\begin{align*}
    {\bf G}=\left(\begin{array}{ccccc}
        1 &1&1&1&1  \\
         0&1&2&3&4 
    \end{array}\right),\; {\bf H}=\left(\begin{array}{cc}
      1   &0  \\
       1  &1\\
       1&2\\
       1&3\\
       1&4
    \end{array}\right).
\end{align*}
Both ${\bf G}$ and ${\bf H}^{\top}$ are generator matrices of GRS codes with the same evaluation points.
The storage is specified as
\begin{equation*}
 (W_{11},\dots, W_{15})=W_1{\bf G}, \quad (W_{21},\dots, W_{25})=W_2{\bf G},\quad (W_{31},\dots, W_{35})=W_3{\bf G}.   
\end{equation*}
where $W_{mn}\in \mathbb{F}_5^{10\times1}$ for $m\in[3],n\in[5]$.

Let $S$, $S'$ and $S''$ be independently and uniformly chosen from the set of all full rank $10\times 10$ matrices over $\mathbb{F}_5$. Label the rows of $S$ as $\{V_i: i\in [10]\}$, the rows of $S'$ as $\{V_i': i\in [10]\}$ and the rows of $S''$ as $\{Z_1, U_1, \dots, U_9\}$. Define
\begin{align*}
\mathcal{V}_1&=\{V_1, V_2, V_3, V_4\},\quad ~ \mathcal{V}_1'=\{V_1', V_2', V_3', V_4'\},\quad ~\mathcal{U}_1=\{Z_1, \widetilde{U}_{1,1}, \widetilde{U}_{1,2},\widetilde{U}_{1,3}\},\\
\mathcal{V}_2&=\{V_1, V_5, V_6,V_7\},\quad ~\mathcal{V}_2'=\{V_1', V_5', V_6',V_7'\},\quad ~\mathcal{U}_2=\{Z_1, \widetilde{U}_{2,1}, \widetilde{U}_{2,2},\widetilde{U}_{2,3}\},\\
\mathcal{V}_3&=\{V_2, V_5, V_8,V_9\},\quad ~ \mathcal{V}_3'=\{V_2', V_5', V_8',V_9'\},\quad ~ \mathcal{U}_3=\{Z_1, \widetilde{U}_{3,1},\widetilde{U}_{3,2},\widetilde{U}_{3,3}\},\\
 \mathcal{V}_4&=\{V_3, V_6, V_8,V_{10}\},\quad \mathcal{V}_4'=\{V_3', V_6', V_8',V_{10}'\},\quad \mathcal{U}_4=\{Z_1, \widetilde{U}_{4,1},\widetilde{U}_{4,2},\widetilde{U}_{4,3}\},\\
\mathcal{V}_5&=\{V_4, V_7, V_9,V_{10}\},\quad \mathcal{V}_5'=\{V_4', V_7', V_9',V_{10}'\},\quad \mathcal{U}_5=\{Z_1, \widetilde{U}_{5,1},\widetilde{U}_{5,2},\widetilde{U}_{5,3}\},
\end{align*}
where 
   \begin{eqnarray*}
\begin{array}{l}
\left(\begin{array}{ccc}\widetilde{U}_{1,1} & \widetilde{U}_{1,2} & \widetilde{U}_{1, 3} \\ 
 
\vdots &\vdots & \vdots\\
\widetilde{U}_{5, 1} & \widetilde{U}_{5,2} & \widetilde{U}_{5, 3}\end{array}\right) 
\triangleq{\bf H}\left(\begin{array}{lll}U_{1} & U_2 & U_{3} \\
U_{4} & U_5 & U_{6}\end{array}\right).
\end{array}
\end{eqnarray*}

Let $\mathcal{P}=\{1,2\}$ be the set of indices of the desired files.
For $n\in[5]$, $Q_n^{[\mathcal{P}]} = (Q_n^{[\mathcal{P}]}(W_1), Q_n^{[\mathcal{P}]}(W_{2}),Q_n^{[\mathcal{P}]}(W_{3}))=(\mathbb{\pi}_{1,n}(\mathcal{V}_n),\mathbb{\pi}_{2,n}(\mathcal{V}_n'),\mathbb{\pi}_{3,n}(\mathcal{U}_n))$,
where $\mathbb{\pi}_{1,n}$, $\mathbb{\pi}_{2,n}$ $\mathbb{\pi}_{3,n}$ are independently and uniformly random permutations in $S_4$. After receiving the queries, the $n^{\mathrm{th}}$ server applies the query vectors to its stored contents $(W_{1n},\, W_{2n},\, W_{3n})$, producing $8$ queried desired symbols $\mathbb{\pi}_{1,n}(\mathcal{V}_n)W_{1n},\mathbb{\pi}_{2,n}(\mathcal{V}_n')W_{2n}$ and 
$4$ queried undesired symbols $\mathbb{\pi}_{3,n}(\mathcal{U}_n)W_{3n}$.

From Example \ref{ex(5,2,2)}, there are $9$ redundant symbols among the queried undesired symbols, i.e., $\{\mathbb{\pi}_{3,n}(\mathcal{U}_n)W_{3n}:n\in[5]\}$ span a linear space of dimension at most $I=11$. Regarding the combination strategy, the $I_n$'s and $C_n$'s are the same as in Example~\ref{ex(5,2,2)}. For $n\in [5]$, a preset deterministic full rank matrix $C_n$ of order $4$ is applied to transform 
\begin{eqnarray*}
 (X_{1,n,1},\dots,X_{1,n,{4}}) &=& (C_n\mathbb{\pi}_{1,n}(\mathcal{V}_n)W_{1n})^{\top},\\ (X_{2,n,1},\dots,X_{2,n,{4}}) &=& (C_n\mathbb{\pi}_{2,n}(\mathcal{V}_n')W_{2n})^{\top},\\(X_{3,n,1},\dots,X_{3,n,{4}}) &=& (C_n\mathbb{\pi}_{3,n}(\mathcal{U}_n)W_{3n})^{\top}.
 \end{eqnarray*}
Pick an MDS matrix \begin{align*}
    {\mathbf D}=\left(\begin{array}{ccc}
        1 & 1&1 \\
        0 &1&2 
    \end{array}
    \right).
\end{align*}
A linear combination function $\mathcal{L}'_n$ is used to compress  $(X_{1,n,1},\dots,X_{1,n,4},X_{2,n,1},\dots,X_{2,n,4},X_{3,n,1},\dots,X_{3,n,4})$ into 
{\it downloaded desired symbols} $\{X_{1,n,1},\dots,X_{1,n,{I_n}},X_{2,n,1},\dots,X_{2,n,{I_n}}\}$, {\it downloaded undesired symbols} $\{X_{3,n,1},\dots,X_{3,n,{I_n}}\}$ and $2\cdot(4-I_n)$ pieces of {\it 3-wise summation symbols} $\left\{ \mathbf{D} (X_{1,n,i+I_n}, X_{2,n,i+I_n},X_{3,n,i+I_n})^{\top}: i\in\left[4-I_n\right]\right\}$. For example, $I_3=3$ and the two summation symbols produced by the $3^{\mathrm{rd}}$ server is $X_{1,3,4}+X_{2,3,4}+X_{3,3,4}$ and $X_{2,3,4}+2X_{3,3,4}$.

Regardless of which files are retrieved, ${\mathbf D}$ is fixed just like the $C_n$'s. Therefore, together with the previous privacy analysis on the queries, the multi-file PIR scheme also guarantees privacy (note that the colluding servers will know the number of desired files, i.e., the value of $P$, based on the dimension of $\mathbf{D}$, but this is allowed in the literature of multi-file PIR). The 11 downloaded undesired symbols $\bigcup_{n\in [5]}{\{X_{3,n,1},\dots,X_{3,n,{I_n}}\}}$ can linearly generate all the queried undesired symbols $\{\mathbb{\pi}_{3,n}(\mathcal{U}_n)W_{3n}:n\in[5]\}$, regardless of the random permutations on the query sets over all servers.
Thus the interferences of the undesired file in $\left\{ \mathbf{D} (X_{1,n,i+I_n}, X_{2,n,i+I_n},X_{3,n,i+I_n})^{\top}: i\in\left[4-I_n\right]\right\}$ can be eliminated. Due to the MDS property of the ${\mathbf D}$, $\{X_{1,n,i+I_n},X_{2,n,i+I_n}, i\in [4-I_n] \}$ can be recovered for each $n\in [5]$. Therefore, the user can collect all the desired symbols $\{\mathbb{\pi}_{1,n}(\mathcal{V}_n)W_{1n},\mathbb{\pi}_{2,n}(\mathcal{V}_n')W_{2n}: n\in [5]\}$ to recover $W_1$ and $W_2$. The rate of the scheme is $\frac{40}{51}$.
\end{example}

\subsection{Cyclically adjacent colluding sets}\label{subsec:cyc}

In certain practical applications, collusion among servers may be restricted. For example, servers that are geographically or logically adjacent are more likely to collude, rather than any two  servers colluding.
In this subsection, for $(M,N,K)=(2,N,K)$, we consider the case where only two cyclically consecutive servers could collude, i.e., the collusion pattern is $\{\{1, 2\}, \{2, 3\}, \dots, \{N-1, N\}, \{N, 1\}\}$. We will show how the rate could be further improved compared to our general scheme for $(M,N,T,K)=(2,N,2,K)$.

\subsubsection{Files and storage code}
Each file is comprised of $L=KN$ independent symbols from a sufficiently large finite field $\mathbb{F}_q$:
\begin{align*}
    W_1=({\bf a}_1, {\bf a}_2,\dots,{\bf a}_K),\quad W_2=({\bf b}_1, {\bf b}_2,\dots,{\bf b}_K),
\end{align*}
where ${\bf a}_1, \dots,{\bf a}_K,{\bf b}_1,\dots, {\bf b}_K\in \mathbb{F}_q^{N\times1}$ are vectors comprised of i.i.d. uniform symbols from $\mathbb{F}_q$. Each file is stored in the system via an $(N,K)$-MDS code with generator matrix ${\bf G}_{K\times N}$, i.e.,
\begin{align*}
(W_{11},\dots, W_{1N})=W_1{\bf G}, \quad (W_{21},\dots, W_{2N})=W_2{\bf G},
\end{align*}
where $W_{mn}\in \mathbb{F}_q^{N\times1}$ for $m\in[2],n\in[N]$.

\subsubsection{The disguising phase}
The user privately chooses two matrices, $S$ and $S'$, independently and uniformly from the set of all full rank $N\times N$ matrices over $\mathbb{F}_q$. Label the rows of $S$ as $\{V_i, i\in [N]\}$. Divide the rows of $S'$ into two parts and label them as $\{Z_i:i\in [K-1]\}$ and $\{U_i:i\in[N-K+1]\}$. 

Define $\mathcal{V}_1,\mathcal{V}_2,\dots,\mathcal{V}_N$ to be the query sets for the desired file as follows. Each $\mathcal{V}_n$ contains $K$ row vectors from $S$. Every $K$ sets with cyclically consecutive indices share a common row vector. As a result, every two cyclically adjacent colluding servers share exactly $K-1$  common row vectors and altogether $K+1$ linearly independent query vectors. 

Next, define $\mathcal{U}_1,\mathcal{U}_2,\dots,\mathcal{U}_N$ to be the query sets for the undesired file. Firstly, we put $\{Z_i:i\in [K-1]\}$ into all the sets $\mathcal{U}_1,\mathcal{U}_2,\dots,\mathcal{U}_N$. Then let 
\begin{equation*}
\left(
  \begin{array}{c}
    \widetilde{U}_{1} \\
    \widetilde{U}_{2} \\
    \vdots\\
    \widetilde{U}_{N}\\
  \end{array}
\right) \triangleq {\bf H} \left(
  \begin{array}{c}
    U_{1}  \\
    U_{2} \\
  \end{array}
\right),
\end{equation*}
and assign $\widetilde{U}_{n}$ to $\mathcal{U}_n$, for $n\in[N]$. Here the matrix $\bf H$ is from \cref{lem:G-H} with the row-MDS property. In this way, every two cyclically adjacent colluding servers share also exactly $K-1$ common query vectors and altogether $K+1$ linearly independent query vectors regarding the undesired file.

Except for the difference in building these query sets, the rest of the disguising phase is the same as before and thus omitted. Finally the $n^{\mathrm{th}}$ server produces $K$ queried desired symbols $\mathbb{\pi}_n(\mathcal{V}_n)W_{mn}$ and $K$ queried undesired symbols $\mathbb{\pi}_n'(\mathcal{U}_n)W_{m^cn}$.

\subsubsection{The squeezing phase}

The squeezing phase works in almost the same way as before. The only thing to note is the number of redundant symbols among the queried undesired symbols: $(K-1)(N-K)$ symbols among $\{Z_iW_{m^cn}:n\in[N],i\in[K-1]\}$ and another $\zeta$ symbols among $\{\widetilde{U}_{n}W_{m^cn}:n\in[N]\}$. Similar as the previous analysis, for arbitrary MDS-coded system we have $\zeta=1$ when $N\leq 2K$ and $\zeta=N-2K$ when $N> 2K$, and for GRS-coded system we have $\zeta=N-K-1$. In total, the queried undesired symbols span a linear space of dimension at most $I\triangleq L-(K-1)(N-K)-\zeta$. 

The rest of the scheme is similar as the general scheme in \cref{Sec:main} and is thus omitted. Then we generalize Theorems \ref{th-1} and \ref{thm:GRS} as follows.

\begin{theorem}\label{thm:cyclic}
For $(M,N,K)=(2,N,K)$ where $N\geq K+2$ and the colluding sets only consist of every two cyclically adjacent servers, there exists a PIR scheme with rate $R=\frac{KN}{KN+K^2-K+N-1}$ when $N\leq 2K$ and $R=\frac{N}{N+K+1}$ when $N>2K$. Moreover, if the storage code is a GRS code, then there exists a PIR scheme with rate $R=\frac{KN}{KN+K^2+1}$. 
\end{theorem}

\begin{example}
Let $(M,N,K)=(2,5,3)$ and the collusion pattern be $\{\{1, 2\}, \{2, 3\}, \{3, 4\}, \{4, 5\},\{5, 1\}\}$. Each file is comprised of $L=15$ independent symbols from $\mathbb{F}_5$:
\begin{align*}
    W_1=({\bf a}_1, {\bf a}_2,{\bf a}_3),\quad W_2=({\bf b}_1, {\bf b}_2,{\bf b}_3),
\end{align*}
where ${\bf a}_1, {\bf a}_2,{\bf a}_3,{\bf b}_1, {\bf b}_2,{\bf b}_3\in \mathbb{F}_5^{5\times1}$. Each file is stored in the system via a $(5,3)$-MDS code. The matrices $\mathbf{G},\mathbf{H}$ and the storage method are the same as in Subsection \ref{(2,5,2,3)}.

The user privately chooses two matrices, $S$ and $S'$, independently and uniformly from the set of all full rank $5\times 5$ matrices over $\mathbb{F}_5$. Label the rows of $S$ as $\{V_i, i\in [5]\}$. Divide the rows of $S'$ into two parts and label them as $\{Z_i:i\in [2]\}$ and $\{U_i:i\in[3]\}$. 
Define
\begin{eqnarray*}
\mathcal{V}_1=\{V_1, V_4,V_5\},&&\mathcal{U}_1=\{Z_1,Z_2,\widetilde{U}_{1}\},\\
\mathcal{V}_2=\{V_1, V_2,V_5\},&&\mathcal{U}_2=\{Z_1,Z_2,\widetilde{U}_{2}\},\\
\mathcal{V}_3=\{V_1, V_2, V_3\},&&\mathcal{U}_3=\{Z_1,Z_2,\widetilde{U}_{3}\},\\
 \mathcal{V}_4=\{V_2, V_3, V_4\},&&\mathcal{U}_4=\{Z_1,Z_2,\widetilde{U}_{4}\},\\
\mathcal{V}_5=\{V_3, V_4, V_5\},&&\mathcal{U}_5=\{Z_1,Z_2,\widetilde{U}_{5}\},
\end{eqnarray*}
where
\begin{equation*}
\left(
  \begin{array}{c}
    \widetilde{U}_{1} \\
    \widetilde{U}_{2} \\
    \vdots\\
    \widetilde{U}_{5}\\
  \end{array}
\right) \triangleq {\bf H} \left(
  \begin{array}{c}
    U_{1}  \\
    U_{2} \\
  \end{array}
\right).
\end{equation*}

Now among the queried undesired symbols there are at least 5 redundant symbols, 2 from $\{Z_1W_{m^cn}:n\in[5]\}$, 2 from $\{Z_2W_{m^cn}:n\in[5]\}$, and another 1 from $\{\widetilde{U}_{n}W_{m^cn}:n\in[5]\}$. Thus the queried undesired symbols span a linear space of dimension at most $I=10$. Set $I_1=3,~I_2=3,~I_3=3,~I_4=1,~I_5=0$ and
\begin{align*}
    C_1=C_2=C_3=C_5=\left(\begin{array}{ccc}
         1 & 0&0 \\
         0&1&0\\
         0&0&1
    \end{array}\right),\quad C_4=\left(\begin{array}{ccc}
          1 & 1&1\\
         0&1&2\\
         0&1&4
    \end{array}\right).
\end{align*}  
Thus, by the similar combination strategy as described before, we can achieve PIR rate $\frac{L}{L+I}=\frac{3}{5}$, which is better than the rate $\frac{10}{17}$ of the $(2,5,2,3)$-scheme in Subsection \ref{(2,5,2,3)}.
\end{example}

We close this subsection by briefly analyzing how Theorem \ref{thm:cyclic} performs better than Theorems \ref{th-1} and \ref{thm:GRS} for the setting where only two cyclically consecutive servers could collude. Essentially, by observing the structure of the query sets $\mathcal{V}$'s for desired files, one can see that the modified scheme of this subsection is in some sense a ``sub-scheme" of our general scheme in Section \ref{Sec:main}. For the purpose of disguising, the structure of the query sets $\mathcal{U}$'s for undesired files is also changed. The trick is that the ratio of redundant symbols among the queried undesired symbols becomes larger, due to the improved ratio of the query vectors of type $Z$ compared to the query vectors of type $\widetilde{U}$. It is expected that the results in this subsection might be further generalized to PIR against arbitrary collusion patterns.
%as long as we can carefully design the structure of $\mathcal{V}$ such that any set of colluding servers have a constant number of common query vectors towards the desired file.

\section{Some converse bounds}\label{sec:converse}

In this section we provide the information theoretic converse bound of MDS-TPIR for two settings. One is the case $(M,N,T,K)=(2,N,2,2)$, for which we show that our scheme for GRS-coded system in Theorem \ref{thm:GRS} achieves the linear PIR capacity. Here the notion ``linear capacity"~\cite{sun2017private} indicates that we only consider linear PIR schemes, in which the responses from servers must be inner products of query vectors and stored message vectors. The other is the case $(M,N,K)=(2,N,K)$ and any two cyclically adjacent servers could collude,  for which we also show that our scheme for GRS-coded system in Theorem \ref{thm:cyclic} achieves the capacity (not necessarily linear).

\subsection{The converse bound for $(M,N,T,K)=(2,N,2,2)$}\label{subsec:con-main}

Without loss of generality, we may always assume that the scheme has symmetry across servers{\footnote{For any scheme designed for a specific file size $L$, we may repeat the scheme for each of the $N!$ permutations of the servers to form a larger scheme designed for file size $N!L$, which has symmetry across servers.}}. Then in a linear PIR scheme, the download from each server is a function of $d\leq L/K=L/2$ linearly independent coded symbols from each file. In other words, to retrieve $W_m$, $m\in[2]$, the responses from the servers can be expressed as:
\begin{equation} \label{An=VW} 
A_n^{[m]}=V_{1n}^{[m]} W_{1n} + V_{2n}^{[m]} W_{2n}, \quad \operatorname{rank}(V_{1n}^{[m]})= \operatorname{rank}(V_{2n}^{[m]}) = d, \quad \forall n \in [N].
\end{equation}
Here, $V_{1n}^{[m]}$ and $V_{2n}^{[m]}$ are  matrices of size $\frac{D}{N} \times \frac{L}{2}$ over $\mathbb{F}_q$, generated randomly by the user, where $D$ represents the total download size. The constraint $\operatorname{rank}(V_{1n}^{[m]})= \operatorname{rank}(V_{2n}^{[m]}) = d$ is a consequence of the privacy requirement. Furthermore, let the notation $\mathbf{V}$ represent the linear space spanned by all row vectors of a matrix $V$ over $\mathbb{F}_q$. By the symmetry of the scheme and the privacy requirement against $T=2$ servers, for $m\in[2]$, the dimension of $\mathbf{V}_{1i}^{[m]}\cap \mathbf{V}_{1j}^{[m]}$ and $\mathbf{V}_{2i}^{[m]}\cap \mathbf{V}_{2j}^{[m]}$ must be equal, and we denote this value as $\dim(\mathbf{V}_{i\cap j})=\alpha d$, where $i,j \in [N], i\neq j$, and $0\leq \alpha \leq 1$. To recover all $L$ symbols of the desired file, it is necessary that $Nd \geq L$. Define $\epsilon \geq 0$ such that
\begin{align}\label{Nd=L}
    Nd = L(1 + \epsilon). 
\end{align}

Now it holds that
\begin{eqnarray}
  D\geq H(A_{1:N}^{[1]} | \mathcal{F}, \mathcal{G}) &\overset{\eqref{correctness}}{=}& H(A_{1:N}^{[1]}, W_1 | \mathcal{F}, \mathcal{G})\notag \\
&\geq& H(W_1 | \mathcal{F},\mathcal{G}) + H(A_1^{[1]}| W_1, \mathcal{F},\mathcal{G})+H(A_{2}^{[1]},A_{3}^{[1]} | W_1,A_1^{[1]},\mathcal{F}, \mathcal{G})\notag\\
&\overset{\eqref{W_{mn}}\eqref{F-W}\eqref{Answer}}{\geq}& H(W_1) + H(A_1^{[1]}| W_1, \mathcal{F},\mathcal{G})+H(A_{2}^{[1]},A_{3}^{[1]} | W_1,W_{21},\mathcal{F}, \mathcal{G})\notag \\
& = &H(W_1) + H(A_1^{[2]}| W_1, \mathcal{F},\mathcal{G})+H(A_{2}^{[2]},A_{3}^{[2]} | W_1,W_{21},\mathcal{F}, \mathcal{G})\label{eqn:15}  \\
&\overset{\eqref{W}\eqref{An=VW}}{=}&L+H(V_{21}^{[2]}W_{21}| W_1,\mathcal{F},\mathcal{G})+H(V_{22}^{[2]}W_{22},V_{23}^{[2]}W_{23}| W_1,W_{21},\mathcal{F},\mathcal{G})\notag\\
&=&L+H(V_{21}^{[2]}W_{21}| \mathcal{F},\mathcal{G})+H(V_{22}^{[2]}W_{22},V_{23}^{[2]}W_{22}| \mathcal{F},\mathcal{G}) \label{eqn:16} \\
&=&L+\dim(\mathbf{V}_{21}^{[2]})+\dim(\mathbf{V}_{22}^{[2]}+\mathbf{V}_{23}^{[2]})\notag\\
&=&L+d+\dim(\mathbf{V}_{22}^{[2]})+\dim(\mathbf{V}_{23}^{[2]})-\dim(\mathbf{V}_{2\cap3})\notag\\
&=&L+d+2d-\alpha d\notag \\
&=&L+(3-\alpha)d\label{L+(3-)d},
\end{eqnarray}
where (\ref{eqn:15}) follows from a more general fact as explained in~\cite[Lemma 2]{sun2017private} (intuitively, it means that the joint entropy of the downloads from up to $T$ servers should be equal regardless of which file is retrieved), and (\ref{eqn:16}) follows from the fact that $W_{23}$ is a linear combination of $W_{21}$ and $W_{22}$.

The next lemma describes a simple yet powerful fact: Whenever a fraction of the downloads regarding the desired file can be shown to be redundant, then its size is upper bounded by $\epsilon L$ since we only download at most $Nd=L(1+\epsilon)$ symbols for the desired file. For simplicity we state this lemma by assuming that $W_2$ is the desired file.

\begin{lemma}\label{lem:sumdim}
For $i\in [3:N]$, let $\mathcal{I}_{i-1}=[i-1]$. If there exists a linear subspace $\mathbf{V}_{i}\subseteq \mathbf{V}_{2i}^{[2]}$ with $H(\mathbf{V}_{i}W_{2i}| V_{2\mathcal{I}_{i-1}}^{[2]}W_{2\mathcal{I}_{i-1}}, \mathcal{F}, \mathcal{G})=0$, then $\sum_{i\in [3:N]}{\dim(\mathbf{V}_{i})}\leq \epsilon L$.
\end{lemma}

\begin{IEEEproof}
    Since $H(\mathbf{V}_{i}W_{2i}| V_{2\mathcal{I}_{i-1}}^{[2]}W_{2\mathcal{I}_{i-1}}, \mathcal{F}, \mathcal{G})=0$, it holds that\footnote{Here by abuse of notation, $V_{2\mathcal{I}_{i-1}}^{[2]}W_{2\mathcal{I}_{i-1}}$ denotes $\{V^{[2]}_{2j}W_{2j}:j\in\mathcal{I}_{i-1}\}$. The same holds for other similar notations.} \begin{equation}\label{3}  H(V_{2i}^{[2]}W_{2i}|\mathbf{V}_{i}W_{2i}, \mathcal{F},\mathcal{G})\geq H(V_{2i}^{[2]}W_{2i}|V_{2\mathcal{I}_{i-1}}^{[2]}W_{2\mathcal{I}_{i-1}},\mathcal{F},\mathcal{G}).
    \end{equation}
Thus, the dimension of $\mathbf{V}_{i}$ can be bounded by
\begin{align}
\dim(\mathbf{V}_{i})=&H(\mathbf{V}_{i}W_{2i}| \mathcal{F},\mathcal{G})\notag\\
=& H(V_{2i}^{[2]}W_{2i}|\mathcal{F},\mathcal{G})-H(V_{2i}^{[2]}W_{2i}|\mathbf{V}_{i}W_{2i}, \mathcal{F},\mathcal{G})\notag\\
\overset{\eqref{3}}{\leq}&
    d-H(V_{2i}^{[2]}W_{2i}| V_{2\mathcal{I}_{i-1} }^{[2]}W_{2\mathcal{I}_{i-1}}, \mathcal{F},\mathcal{G})\label{2}.
    \end{align}
Moreover,   
\begin{align}
    L&=H(A_{1:N}^{[2]}| W_1,\mathcal{F},\mathcal{G})\notag\\
    &=H(V_{21}^{[2]}W_{21}, \ldots,V_{2N}^{[2]}W_{2N}| \mathcal{F},\mathcal{G})\notag\\
     &= H(V_{21}^{[2]}W_{21},V_{22}^{[2]}W_{22}| \mathcal{F}, \mathcal{G})+\sum_{i\in [3:N] }H(V_{2i}^{[2]}W_{2i}| V_{2\mathcal{I}_{i-1} }^{[2]}W_{2\mathcal{I}_{i-1}}, \mathcal{F},\mathcal{G})\notag\\
    &= \sum_{i\in [1:2] }H(V_{2i}^{[2]}W_{2i}| \mathcal{F}, \mathcal{G})+\sum_{i\in [3:N] }H(V_{2i}^{[2]}W_{2i}| V_{2\mathcal{I}_{i-1} }^{[2]}W_{2\mathcal{I}_{i-1}}, \mathcal{F},\mathcal{G})\label{4}\\
    %&\quad +H(V_{24}^{[2]}W_{24}| V_{21}^{[2]}W_{21}, V_{22}^{[2]}W_{22},\mathcal{F},\mathcal{G})+H(V_{25}^{[2]}W_{25}| V_{21}^{[2]}W_{21}, V_{22}^{[2]}W_{22},\mathcal{F},\mathcal{G})\notag\\
    &=2d+\sum_{i\in [3:N] }H(V_{2i}^{[2]}W_{2i}| V_{2\mathcal{I}_{i-1} }^{[2]}W_{2\mathcal{I}_{i-1}}, \mathcal{F},\mathcal{G}),\label{1}
\end{align}
where \eqref{4} follows from  the independence of the contents on two servers from an $(N,2)$-MDS coded storage system. Combining \eqref{2} and \eqref{1}, it holds that $\sum_{i\in [3:N]}\dim(\mathbf{V}_{i}) \leq Nd-L=\epsilon L$. 
\end{IEEEproof}

To fully exploit Lemma \ref{lem:sumdim}, we pick the proper linear subspaces $\mathbf{V}_i\subseteq\mathbf{V}^{[2]}_{2i}$ as follows.

\begin{lemma}\label{lem:red}
 For $i\in [3:N]$, let $\mathcal{I}_{i-1}=[i-1]$. Let $\mathbf{V}_i=\left(\sum_{j=2}^{i-1}\mathbf{V}_{1\cap j }\right)\cap \mathbf{V}_{1\cap i }$,  then $H(\mathbf{V}_{i}W_{2i}| V_{2\mathcal{I}_{i-1}}^{[2]}W_{2\mathcal{I}_{i-1}}, \mathcal{F}, \mathcal{G})=0$.
\end{lemma}
\begin{IEEEproof}
For any vector $\mathbf{v} \in \mathbf{V}_i=\left(\sum_{j=2}^{i-1}\mathbf{V}_{1\cap j}\right)\cap \mathbf{V}_{1\cap i }$, there exist $\mathbf{v}_j\in \mathbf{V}_{1\cap j}$ for $j\in[2:i-1]$ such that $\mathbf{v}=\sum_{j=2}^{i-1}\mathbf{v}_j$. Moreover, by the MDS property \eqref{mds}, we know that for all $j\in[2:i-1]$, $W_{2i}=a_{i,j}W_{21}+b_{i,j}W_{2j}$ for some certain coefficients $a_{i,j},b_{i,j}\in\mathbb{F}_q^*$. Now it holds that
\begin{equation*}
\mathbf{v}W_{2i}=\sum_{j=2}^{i-1}\mathbf{v}_jW_{2i}
=\sum_{j=2}^{i-1}\mathbf{v}_j\left(a_{i,j}W_{21}+b_{i,j}W_{2j}\right)
=\sum_{j=2}^{i-1}a_{i,j}\mathbf{v}_jW_{21}+\sum_{j=2}^{i-1}b_{i,j}\mathbf{v}_jW_{2j}.
\end{equation*}
On one hand, since $\mathbf{v}_j\in\mathbf{V}^{[2]}_{21}$, $\mathbf{v}_jW_{21}$ is a linear combination of $\{V^{[2]}_{21}W_{21}\}$. On the other hand, since $\mathbf{v}_j\in\mathbf{V}^{[2]}_{2j}$, $\mathbf{v}_jW_{2j}$ is a linear combination of $\{V^{[2]}_{2j}W_{2j}\}$, $j\in[2:i-1]$. Thus, $\mathbf{v}W_{2i}$ is a linear combination of 
$\{V^{[2]}_{2j}W_{2j}:j\in [1:i-1]\}$ and hence
$H(\mathbf{V}_{i}W_{2i}| V_{2\mathcal{I}_{i-1}}^{[2]}W_{2\mathcal{I}_{i-1}}, \mathcal{F}, \mathcal{G})=0$.
\end{IEEEproof}
According to Lemma \ref{lem:sumdim} and Lemma \ref{lem:red}, one can directly obtain the following corollary:
\begin{corollary}\label{cor1} Following the previous notations, it holds that
    \begin{equation}
\sum_{i={3}}^{N} \dim \left(\left(\sum_{j=2}^{i-1}\mathbf{V}_{1\cap j }\right)\cap \mathbf{V}_{1\cap i}\right)\leq \epsilon L. 
    \end{equation}
\end{corollary}
Based on this inequality we are now ready to derive a relationship between $\alpha$ and $\epsilon$. 

\begin{lemma}\label{alpha,d}
Following the previous notations, it holds that $\alpha \leq \frac{1}{N-1}+\frac{N}{N-1}\cdot \frac{\epsilon}{1+\epsilon}$.
   %\begin{align}
    %24\alpha d&\leq 6d+11\epsilon L\\
  %\text{i.e.}, \quad   \alpha &\leq \frac{1}{4}+\frac{55}{24}\frac{\epsilon}%{1+\epsilon}\label{alpha}
%    \end{align}
\end{lemma}
\begin{IEEEproof}
Note that $\sum_{i=2}^{N}\mathbf{V}_{1\cap i}\subseteq \mathbf{V}_{21}^{[2]}$. Then it holds that
\begin{eqnarray}
d&=&\dim\left(\mathbf{V}_{21}^{[2]}\right)\geq \dim\left(\sum_{i=2}^{N}\mathbf{V}_{1\cap i}\right)
\notag\\
&=&\sum_{i=2}^{N}\dim\left(\mathbf{V}_{1\cap i}\right)-\sum_{i=3}^{N}{\dim\left(\left(\sum_{j=2}^{i-1}\mathbf{V}_{1\cap j }\right)\cap \mathbf{V}_{1\cap i}\right)}\notag\\
&\geq& (N-1)\alpha d-\epsilon L\notag.
\end{eqnarray}

Plugging in $L=\frac{Nd}{1+\epsilon}$, we derive that $\alpha \leq \frac{1}{N-1}+\frac{N}{N-1}\cdot \frac{\epsilon}{1+\epsilon}$.
\end{IEEEproof}

According to Lemma \ref{alpha,d}, we obtain 
\begin{align*}
   H(A_{1:N}^{[1]} | \mathcal{F}, \mathcal{G})&\overset{\eqref{L+(3-)d}}{\geq} L+(3-\alpha)d \\
   &\geq L+\left(3-\frac{1}{N-1}-\frac{N}{N-1}\cdot \frac{\epsilon}{1+\epsilon}\right)\frac{1+\epsilon}{N}L\\
   &=\left(1+\frac{3N-4}{N(N-1)}\right)L+\left(\frac{3N-4}{N}-1\right)\frac{\epsilon L}{N-1}.\\
   &\geq \frac{N^2+2N-4}{N^2-N}L
\end{align*}
Thus, the linear MDS-TPIR capacity for the parameters $(2,N,2,2)$ is at most $\frac{N^2-N}{N^2+2N-4}$, and if equality holds then $\epsilon$ will be forced to be zero. Recall that the rate in \cref{thm:GRS} achieved by our scheme for GRS-coded system is indeed $R(2,N,2,2)=\frac{N^2 - N}{N^2 + 2N - 4}$. Thus, we obtain the following result.

\begin{theorem}
For $(M,N,T,K)=(2,N,2,2)$ and $N\geq 4$, the linear MDS-TPIR capacity is $\frac{N^2 - N}{N^2 + 2N - 4}$.
\end{theorem}

\subsection{The converse bound for $(M,N,K)=(2,N,K)$ where two cyclically adjacent servers collude} \label{subsec:con-cyc}

In this subsection, for $(M,N,K)=(2,N,K)$ where only any two cyclically adjacent servers could collude, we can prove that our scheme in Theorem \ref{thm:cyclic} does achieve the capacity (not necessarily linear).

\begin{theorem}
For $(M,N,K)=(2,N,K)$ where two cyclically adjacent servers collude, the capacity is $C= \frac{NK}{NK+K^2+1}$.
\end{theorem}
\begin{IEEEproof}For $j\in[K:N-1]$, we have\footnote{Denote $A_{1:K-1}^{[1]}=\{A_n^{[1]}: n\in[1:K-1]\}$ and $W_{2,1:K-1}=\{W_{2n}:n\in [1:K-1]\}.$ The same holds for other similar notations.}
\begin{eqnarray}
H(A_{1:N}^{[1]} | \mathcal{F}, \mathcal{G}) &\overset{\eqref{correctness}}{=}& H(A_{1:N}^{[1]}, W_1 | \mathcal{F}, \mathcal{G}) \notag\\
&\geq &H(W_1| \mathcal{F},\mathcal{G}) + H(A_{1:K-1}^{[1]}| W_1, \mathcal{F},\mathcal{G}) + H(A_{j}^{[1]},A_{j+1}^{[1]} | W_1, A_{1: K-1}^{[1]}, \mathcal{F}, \mathcal{G})\notag\\
&=&H(W_1) + \sum_{i=1}^{K-1}H(A_i^{[1]}| W_1, \mathcal{F},\mathcal{G})+ H(A_{j}^{[1]},A_{j+1}^{[1]} | W_1, A_{1: K-1}^{[1]},\mathcal{F}, \mathcal{G})\label{eqn:24}\\
&\overset{\eqref{W}\eqref{W_{mn}}\eqref{Answer}}{\geq}& L + \sum_{i=1}^{K-1}H(A_i^{[1]}| W_1, \mathcal{F},\mathcal{G})+ H(A_{j}^{[1]},A_{j+1}^{[1]} | W_1, W_{2,1:K-1},\mathcal{F}, \mathcal{G})\notag\\
&=&L + \sum_{i=1}^{K-1}H(A_i^{[2]}| W_1, \mathcal{F},\mathcal{G})+H(A_{j}^{[2]},A_{j+1}^{[2]} | W_1, W_{2,1:K-1},\mathcal{F}, \mathcal{G}),\label{eq:capacity}
\end{eqnarray}
where \eqref{eqn:24} follows from \eqref{F-W} and the independence of the contents on $K-1$ servers from an $(N,K)$-MDS coded storage system, and \eqref{eq:capacity} follows from a more general fact as explained in \cite[Lemma 2]{sun2017private}. Summing up the last term in Inequality \eqref{eq:capacity} for $j\in [K,N-1]$, we obtain
\begin{align}
\sum^{N-1}_{j=K}H(A_{j}^{[2]},A_{j+1}^{[2]} | W_1, W_{2,1:K-1},\mathcal{F}, \mathcal{G})\geq&H(A_{K:N}^{[2]} | W_1, W_{2,1:K-1},\mathcal{F}, \mathcal{G})+\sum_{j=K+1}^{N-1}H(A_{j}^{[2]} | W_1, W_{2,1:K-1},\mathcal{F}, \mathcal{G}) \label{eqn:26}\\
=&H(A_{K:N}^{[2]} | W_1, W_{2,1:K-1},\mathcal{F}, \mathcal{G})+\sum_{j=K+1}^{N-1}H(A_{j}^{[2]} | W_1,\mathcal{F}, \mathcal{G})\label{eqn:27}\\
=&\frac{L}{K}+\sum_{j=K+1}^{N-1}H(A_{j}^{[2]} | W_1,\mathcal{F}, \mathcal{G}).\label{eqn:28}
\end{align}
Here \eqref{eqn:26} follows from repeatedly performing the information inequality of the form $H(A,B)+H(B.C)\geq H(A,B,C)+H(B)$, \eqref{eqn:27} results from the independence of $W_{2,j}$ with $W_{2,1:K-1}$ in an $(N,K)$-MDS storage system for every $j\in[K+1:N-1]$, and \eqref{eqn:28} holds since $H(A_{K:N}^{[2]} | W_1, W_{2,1:K-1},\mathcal{F}, \mathcal{G})\overset{\eqref{Answer}}{=}H(A_{1:N}^{[2]} | W_1, W_{2,1:K-1},\mathcal{F}, \mathcal{G})\overset{\eqref{correctness}}{=}H(W_2 | W_1, W_{2,1:K-1},\mathcal{F}, \mathcal{G})=L/K$.

Combining  \eqref{eq:capacity} and \eqref{eqn:28}, it holds that 
\begin{equation}\label{eq:capacity2}
    H(A_{1:N}^{[1]} | \mathcal{F}, \mathcal{G})\geq  L + \sum_{i=1}^{K-1}H(A_i^{[2]}| W_1, \mathcal{F},\mathcal{G}) +\frac{1}{N-K}\left(\frac{L}{K}+\sum_{j=K+1}^{N-1}H(A_{j}^{[2]}| W_1,\mathcal{F}, \mathcal{G})\right).
\end{equation}

Now consider the permutation $\pi\in S_N$ where $\pi(i)=i+1$ for $i\in[N-1]$ and $\pi(N)=1$. Let $\pi^s$ denote the $s$-th power of $\pi$ and in particular $\pi^0$ is the identity permutation. By cyclically shifting the subscripts via $\pi^s$ on \eqref{eq:capacity2}, we have the following $N$ inequalities. For $0\leq s\leq N-1$,
\begin{align*}
H(A_{1:N}^{[1]} | \mathcal{F}, \mathcal{G})
    &\geq L + \sum^{K-1}_{i=1}H(A_{\pi^s(i)}^{[2]}| W_1, \mathcal{F},\mathcal{G}) 
 +\frac{1}{N-K}\left(\frac{L}{K}+\sum_{i=K+1}^{N-1}H(A_{\pi^s(i)}^{[2]} | W_1,\mathcal{F}, \mathcal{G})\right).
\end{align*}
Summing up the $N$ inequalities above, it holds that
\begin{align*}
NH(A_{1:N}^{[1]} | \mathcal{F}, \mathcal{G})
    &\geq NL + (K-1)\sum_{i=1}^{N} H(A_i^{[2]}| W_1, \mathcal{F},\mathcal{G})\\
&\quad +\frac{N}{N-K}\cdot\frac{L}{K}+\frac{N-K-1}{N-K}\sum_{i=1}^{N} H(A_i^{[2]}| W_1, \mathcal{F},\mathcal{G})\\
&\geq NL + \frac{KN-K^2-1}{N-K} H(A_{1:N}^{[2]}| W_1, \mathcal{F},\mathcal{G})
+\frac{N}{N-K}\cdot\frac{L}{K}\\
&=NL+\frac{KN-K^2-1}{N-K}\cdot L+\frac{N}{N-K}\cdot\frac{L}{K}\\
&=\frac{KN+K^2+1}{K}\cdot L.
\end{align*}
Dividing both sides of the equation above by $NL$, we finally obtain 
\begin{equation*}
    C\leq \frac{L}{H(A_{1:N}^{[1]} | \mathcal{F}, \mathcal{G})}\leq \frac{KN}{KN+K^2+1}.    
\end{equation*}

This upper bound, together with our scheme in Theorem~\ref{thm:cyclic}, determines that the capacity for the case $(M,N,K)=(2,N,K)$ where two cyclically adjacent servers could collude is exactly $C=\frac{KN}{KN+K^2+1}$.
\end{IEEEproof}

\begin{remark}
Note that for the MDS-TPIR model with $(M,N,T,K)=(2,N,2,K)$, since the privacy requirement is stronger than the model where only two cyclically adjacent servers could collude, the capacity $C(2,N,2,K)$ is then upper bounded by $\frac{KN}{KN+K^2+1}$. Recall that for $(M,N,T,K)=(2,N,2,K)$, $N\geq K+2$, Theorem \ref{thm:GRS} presents a scheme for the GRS-coded system with rate $\frac{N^2-N}{N^2+KN-2K}$, which is then a constructive lower bound of the capacity. Thus, the exact value of $C(2,N,2,K)$ is now between $\frac{N^2-N}{N^2+KN-2K}$ and $\frac{KN}{KN+K^2+1}$.  
\end{remark}

\section{When more servers collude}\label{Sec:Tgeq3}

Up until now in this paper we have fixed $T=2$. The generalization of our disguise-and-squeeze scheme to $T\geq 3$ is much more complicated. On one hand, in the disguising phase we will demonstrate that it is not enough to simply design query {\it sets} for each file due to the privacy requirement. Instead, we need to design query {\it spaces} for each file with certain properties for privacy. On the other hand, due to the form of the queries (say, the reduced row echelon form of the query spaces), to find a precise combination strategy in the squeezing phase becomes more difficult or even impossible. Thus, we may have to follow a similar approach as in~\cite{sun2017private} to employ a randomized combination strategy, at the cost of accepting a nonzero error probability for the retrieval correctness. 

In this section, we will start with an illustrative example for $(M,N,T,K)=(2,6,3,3)$, then present the scheme for $T=3$, and finally move on to discuss the general case with arbitrary $T$.

\subsection{An example for $(M,N,T,K)=(2,6,3,3)$}
 
\subsubsection{Files and storage code} Each file is comprised of $L=60$ independent symbols from a sufficiently large finite field $\mathbb{F}_q$:
\begin{align*}
    W_1=({\bf a}_1, {\bf a}_2,{\bf a}_3),\quad W_2=({\bf b}_1, {\bf b}_2,{\bf b}_3),
\end{align*}
where ${\bf a}_1, {\bf a}_2,{\bf a}_3, {\bf b}_1, {\bf b}_2,{\bf b}_3\in \mathbb{F}_q^{20\times1}$. Each file is stored in the system via a $(6,3)$-MDS code. Let 
\begin{align*}
    {\bf G}=\left(\begin{array}{cccccc}
        1 &0&0&1&1&1  \\
         0&1&0&1&2&4\\
         0&0&1&1&3&2
    \end{array}\right),\; {\bf H}=\left(\begin{array}{ccc}
      1   &1&1 \\
       1&2&4\\
       1&3&2\\
       1&0&0\\
       0&1&0\\
       0&0&1
    \end{array}\right).
\end{align*}
Both ${\bf G}$ and ${\bf H}$ are MDS over a finite field of characteristic $7$ (so $q$ is a power of $7$), and thus our protocol is also built over this field.
The storage is specified as
\begin{equation*}
(W_{11},\dots, W_{16})=W_1{\bf G},\quad (W_{21},\dots, W_{26})=W_2{\bf G},   
\end{equation*}
where $W_{mn}\in \mathbb{F}_q^{20\times1}$ for $m\in[2],n\in[6]$.

\subsubsection{The disguising phase}
Let $S$ and $S'$ be independently and uniformly chosen from the set of all full rank $20\times 20$ matrices. Label the rows of $S$ as $\{V_i: i\in [20]\}$. Label the rows of $S'$ as $\{Z_1\} \cup \{U_{1,j},U_{2,j},U_{3,j}:j\in[3]\} \cup\{U_\ell:\ell\in [10]\}$.
Let $\mathcal{V}_1, \mathcal{V}_2, \dots, \mathcal{V}_6$ be the query sets for the desired file, where each $\mathcal{V}_n$ contains $10$ row vectors from $S$, and every $3$ sets out of $\mathcal{V}_1, \mathcal{V}_2, \dots, \mathcal{V}_6$ share exactly one common row vector from $S$. 

Now here comes the vital difference between the schemes for $T=2$ and $T=3$. For $T=3$ we have to analyze the perspective of any three colluding servers. Except for one common vector shared by the three colluding servers, every two out of the three colluding servers share another 3 vectors, which do not appear in the third server. Can we set the the query sets $\mathcal{U}_n$'s for the undesired file  imitating the same properties as the $\mathcal{V}_n$'s, and at the same time aiming for more redundancies among the queried undesired symbols? Setting a common vector $Z_1$ in $\{\mathcal{U}_n:1\leq n\leq 6\}$ is a straightforward way to account for the common vector shared by any three colluding servers. However, now every row vector of $S'$ other than $Z_1$ can only appear at most twice in the $\mathcal{U}_n$'s, or otherwise some three colluding servers will share two common vectors, leading to a contradiction to the property of the $\mathcal{V}_n$'s. As a result, in the $\mathcal{U}_n$'s we will need at least $\frac{54}{2}=27$ query vectors other than $Z_1$, but apparently $S'$ cannot provide enough vectors. Thus, instead of imitating the properties of the $\mathcal{V}_n$'s as {\it sets}, we turn to design the $\mathcal{U}_n$'s to imitate the properties of the $\mathcal{V}_n$'s as {\it linear spaces}. 

Let $\Gamma_{n}$ be the linear space spanned by the query vectors in $\mathcal{V}_{n}$, for $n\in [6]$. The properties of the $\mathcal{V}_n$'s as sets naturally transform to the following properties as linear spaces: For any $\{n_1,n_2,n_3\}\subseteq[6]$, it holds that $\mathrm{dim}(\Gamma_{n_1}\cap\Gamma_{n_2}\cap\Gamma_{n_3})=1$, $\mathrm{dim}\left((\Gamma_{n_1}\cap\Gamma_{n_2})\setminus\Gamma_{n_3}\right)=3$, $\mathrm{dim}\left((\Gamma_{n_1}\cap\Gamma_{n_3})\setminus\Gamma_{n_2}\right)=3$, and $\mathrm{dim}\left((\Gamma_{n_2}\cap\Gamma_{n_3})\setminus\Gamma_{n_1}\right)=3$. Let $\Gamma'_{n}$ be the linear space spanned by the query vectors in $\mathcal{U}_{n}$, for $n\in [6]$. To imitate the same properties of the $\mathcal{V}_n$'s as linear spaces, we define $\mathcal{U}_1, \dots, \mathcal{U}_6$ by
$\mathcal{U}_{n}=\{Z_{1},U_{2n-1,1}^{*}, U_{2n, 1}^{*},U_{2n-1,2}^{*},U_{2n, 2}^{*},U_{2n-1,3}^*,U_{2n, 3}^{*}, \widetilde{U}_{n,1}, \widetilde{U}_{n,2},\widetilde{U}_{n, 3}\}.$ 
Here the common vector $\{Z_1\}$ accounts for the $1$-dimensional intersection for any three query spaces. For $j\in [3]$ and $n\in[6]$, the query vectors $\{U_{2n-1,j}^{*}, U_{2n, j}^{*}\}$ in $\mathcal{U}_n$ are defined by 
\begin{eqnarray*}
    (U_{1,j}^*,U_{2,j}^*,\dots, U_{12,j}^*)^{\top}=\left(
    \begin{array}{rrr}
    0&1&0\\
    0&0&1\\\hdashline
    6&1&0\\
    6&0&1\\\hdashline
    5&1&0\\
    3&0&1\\\hdashline
    4&1&0\\
    5&0&1\\\hdashline
    3&1&0\\
    5&0&1\\\hdashline
    1&0&0\\
    0&1&0
    \end{array}\right)
    \left(\begin{array}{c}
     U_{1,j}  \\
     U_{2,j}\\
     U_{3,j}
\end{array}\right).
\end{eqnarray*}
This matrix is defined over a finite field of characteristic $7$ (i.e., $q$ is a power of $7$).  These vectors are used to guarantee that $\mathrm{dim}\left((\Gamma'_{n_1}\cap\Gamma'_{n_2})\setminus\Gamma'_{n_3}\right)=3$ for any $n_1,n_2,n_3$. Essentially, for each $j\in [3]$, each server will receive a $2$-dimensional subspace of $\mathrm{span}\{U_{1,j},U_{2,j},U_{3,j}\}$, and the $1$-dimensional intersection of every two such subspaces does not lie in another subspace. Finally, the vectors $\{\widetilde{U}_{n,j}:n\in [6],j\in[3]\}$ are defined in the same way as in previous sections:
\begin{eqnarray*}
\begin{array}{l}
\left(\begin{array}{ccc}\widetilde{U}_{1,1} & \widetilde{U}_{1,2} & \widetilde{U}_{1, 3} \\ 
\widetilde{U}_{2,1} & \widetilde{U}_{2,2} & \widetilde{U}_{2, 3} \\ 
\vdots &  \vdots& \vdots \\
\widetilde{U}_{6, 1} & \widetilde{U}_{6,2} & \widetilde{U}_{6, 3 }\end{array}\right)
\triangleq{\bf H}\left(\begin{array}{lll}U_{1} & U_2 & U_{3} \\
U_{4} & U_5 & U_{6}\\
U_{7}& U_8 & U_{9}
\end{array}\right).
\end{array}
\end{eqnarray*}

To send a query space  to a server, the user could in fact pick an arbitrary basis of the space. One natural way, as suggested in~\cite{sun2017private}, is to send the matrix $\mathbb{B}(\mathcal{V})$ of size $10 \times 20$ which is the reduced row echelon form of $\mathcal{V}$. Note that for a given linear space, the reduced row echelon form is unique. Let $W_m$ be the desired file and  $W_{m^c}$ be the undesired file. For $n\in[6]$, let $Q_n^{[m]} =(Q_n^{[m]}(W_m), Q_n^{[m]}(W_{m^c}))=(\mathbb{B} (\mathcal{V}_n),\mathbb{B} (\mathcal{U}_n))$. After receiving the queries, the $n^{\mathrm{th}}$ server applies the query vectors to its stored contents $(W_{mn},\, W_{m^c n})$, producing $10$ {\it queried desired symbols} $\mathbb{B}(\mathcal{V}_n)W_{mn}$ and $10$ {\it queried undesired symbols} $\mathbb{B}(\mathcal{U}_n)W_{m^cn}$.

\subsubsection{The squeezing phase}

Now we analyze the number of redundant symbols among $\{\mathcal{U}_nW_{m^cn}:n\in[6]\}$. Similar as the analysis in the previous sections, there are $N-K=3$ redundant symbols among $\{Z_1W_{m^cn}:n\in[6]\}$, and at least $1$ redundant symbol among $\{\widetilde{U}_{n, j}W_{m^cn}:n\in[6]\}$ for every $j\in[3]$ due to the relations of $\bf G$ and $\bf H$ based on Lemma \ref{lem:G-H}. Moreover, let $W_{m^c}=(\mathbf{x}_1,\mathbf{x}_2,\mathbf{x}_3)$, for every $j\in[3]$, since 
\begin{align*}
    \{{U}_{2n-1, j}^*W_{m^cn}, {U}_{2n, j}^*W_{m^cn}:n\in[6]\}\subseteq \text{span}\{U_{1,j}{\bf x}_1,U_{2,j}{\bf x}_1,U_{3,j}{\bf x}_1,\dots,U_{1,j}{\bf x}_3,U_{2,j}{\bf x}_3,U_{3,j}{\bf x}_3\},
\end{align*}
there are naturally at least $12-9=3$ redundant symbols among $\{{U}_{2n-1, j}^*W_{m^cn}, {U}_{2n, j}^*W_{m^cn}:n\in[6]\}$. Therefore, in total there are at least $3+3\times 1+3\times 3=15$ redundant symbols among $\{\mathcal{U}_nW_{m^cn}:n\in[6]\}$.
Ideally, a carefully designed combination strategy would allow us to download $45$  desired symbols, $45$  undesired symbols, and $15$ paired-up summation symbols. As long as the $45$ downloaded undesired symbols span 
$\{\mathcal{U}_nW_{m^cn}:n\in[6]\}$, the interferences of the undesired file in the paired-up summations could be eliminated and the user successfully retrieves all the queried desired symbols. This scheme has rate $R=\frac{60}{60+60-15}=\frac{4}{7}$, beating the FGHK-conjectured value $\frac{6}{11}$.

\subsection{The scheme for $T=3$}

In the rest of this section, for convenience we assume the system is based on a GRS code to simplify the redundancy analysis in the squeezing phase. Assume that the storage code is $\mathrm{GRS}_{K}(\boldsymbol{\alpha}, \bf{v})$ with generator matrix ${\bf G}$, and the corresponding matrix by Lemma \ref{GRS} is chosen such that ${\bf H}^{\top}$ is the generator matrix of $\mathrm{GRS}_{T}(\boldsymbol{\alpha}, \bf{w})$. 

Assume that each file consists of $L=\tbinom{N}{K}K=\tbinom{N-1}{K-1}N$ symbols from a sufficiently large finite field $\mathbb{F}_q$. Let 
\begin{align*}
W_1 = ({\bf a}_1, \dots, {\bf a}_{K}), \quad W_2 = ({\bf b}_1, \dots, {\bf b}_{K}),
\end{align*}
where ${\bf a}_1,\dots,{\bf a}_{K},{\bf b}_1,\dots,{\bf b}_{K} \in \mathbb{F}_q^{{N\choose K}\times 1}$ are vectors comprised of i.i.d. uniform symbols from $\mathbb{F}_q$. Each file is stored in the system via the same  storage code  $\mathrm{GRS}_{K}(\boldsymbol{\alpha}, \bf{v})$ with generator matrix ${\bf G}_{K\times N}$, i.e.,
\begin{align*}
(W_{11},\dots, W_{1N})=W_1{\bf G}, \quad (W_{21},\dots, W_{2N})=W_2{\bf G},
\end{align*}
where $W_{mn}\in \mathbb{F}_q^{{N\choose K}\times1}$ for $m\in[2],n\in[N]$.

\subsubsection{The disguising phase (Constructions of queries)} Independently choose two random matrices $S, S'$ uniformly from all the  $\tbinom{N}{K}\times \tbinom{N}{K}$ full rank matrices over $\mathbb{F}_q$. Let $\delta=\tbinom{N-3}{K-3}$, $\tau=\tbinom{N-3}{K-2}$, and $\gamma=\tbinom{N-3}{K-1}$. 
Label the rows of $S$ as $\left\{V_i, i\in \left[\tbinom{N}{K}\right]\right\}$. Divide the rows of $S'$ into three parts  $\left\{Z_i:i\in [\delta]\right\}$, $\left\{U_{1,j},U_{2,j},U_{3,j}:j\in[\tau]\right\}$, and $\left\{U_\ell:\ell\in \left[{N \choose K}-\delta-3\tau\right]\right\}$.

The query sets $\mathcal{V}_1,\mathcal{V}_2,\dots,\mathcal{V}_N$ for the desired file are defined in the same way as before, i.e., each $\mathcal{V}_n$ contains $\tbinom{N-1}{K-1}$ row vectors from $S$, and every $K$ sets out of $\mathcal{V}_1,\mathcal{V}_2,\dots,\mathcal{V}_N$ share exactly one common row vector from $S$. Let $\Gamma_{n}$ be the linear space spanned by the query vectors in $\mathcal{V}_{n}$, $n\in [N]$. Then for any $\{n_1,n_2,n_3\}\subseteq[N]$, it holds that $\mathrm{dim}(\Gamma_{n_1}\cap\Gamma_{n_2}\cap\Gamma_{n_3})=\delta$, $\mathrm{dim}\left((\Gamma_{n_1}\cap\Gamma_{n_2})\setminus\Gamma_{n_3}\right)=\tau$, $\mathrm{dim}\left((\Gamma_{n_1}\cap\Gamma_{n_3})\setminus\Gamma_{n_2}\right)=\tau$, and $\mathrm{dim}\left((\Gamma_{n_2}\cap\Gamma_{n_3})\setminus\Gamma_{n_1}\right)=\tau$. 
For disguising, these properties should be imitated in the queries regarding the undesired file. Towards this goal,
we define the query sets $\mathcal{U}_1,\mathcal{U}_2,\dots,\mathcal{U}_N$ for the undesired file as follows. For $n\in [1:N]$,
\begin{eqnarray*}
  \mathcal{U}_{n}=\{Z_{1} , \dots , Z_{\delta}, U_{2n-1,1}^{*} , U_{2n, 1}^{*} , \dots , U_{2n-1,\tau}^*, U_{2n, \tau}^{*}, 
\widetilde{U}_{n,1}, \dots,\widetilde{U}_{n, \gamma}\}. 
\end{eqnarray*}
Let $\Gamma'_{n}$ be the linear space spanned by the vectors in $\mathcal{U}_{n}$, $n\in [N]$. The roles of these vectors are as follows.
\begin{itemize}
    \item $\{Z_1,\dots,Z_{\delta}\}$ are directly picked from $S'$, accounting for the $\delta$-dimensional intersection for any three spaces. 
    \item The vectors $\{\widetilde{U}_{n,j}:{n\in[N],j\in[\gamma]}\}$ are defined by
\begin{eqnarray*}
\left(\begin{array}{ccc}\widetilde{U}_{1,1} & \cdots & \widetilde{U}_{1, \gamma} \\ 
\widetilde{U}_{2,1} & \cdots & \widetilde{U}_{2, \gamma} \\ 
\vdots &  \ddots& \vdots \\
\widetilde{U}_{N, 1} & \cdots & \widetilde{U}_{N, \gamma }\end{array}\right)
={\bf H}\left(\begin{array}{lll}U_{1} & \cdots & U_{\gamma} \\
U_{\gamma +1} & \cdots & U_{2\gamma}\\
U_{2\gamma+1}& \cdots & U_{3\gamma}
\end{array}\right),
\end{eqnarray*} 
and this equation is well defined since it is routine to check that ${N\choose K}-\delta-3\tau \geq 3\gamma$.
    \item The vectors $\{U_{2n-1,j}^*,U_{2n,j}^*:{n\in[N],j\in[\tau]}\}$ are used to guarantee that $\mathrm{dim}\left((\Gamma'_{n_1}\cap\Gamma'_{n_2})\setminus\Gamma'_{n_3}\right)=\tau$ for any $\{n_1,n_2,n_3\}\subseteq[N]$. For $j \in [\tau]$, we define
\begin{eqnarray*}
    (U_{1,j}^* ,U_{2,j}^*, \dots, U_{2N,j}^*)^{\top}={\bf H}^*\left(\begin{array}{c}
     U_{1,j}  \\
     U_{2,j}\\
     U_{3,j}
\end{array}\right),
\end{eqnarray*}
where the matrix ${\bf H}^*$ is of size $2N\times 3$ satisfying the following conditions: 
\begin{itemize}
     \item[{P1:}] For any $\mathcal{T}'=\{n_1,n_2\}\subseteq[N]$, 
     $\mathrm{span}\{U_{2n_1-1,j}^*,U_{2n_1,j}^*\}\cap\mathrm{span}\{U_{2n_2-1,j}^*,U_{2n_2,j}^*\}\triangleq\mathrm{span}\{\mathbf{u}_{\mathcal{T}',j}\}$ has dimension 1. 
     %and we denote this dimension as $U_{n_1,n_2;j}^*$.
    \item[{ P2:}] For any $\mathcal{T}=\{n_1,n_2,n_3\} \subseteq [N]$, $\{\mathbf{u}_{\mathcal{T}',j}:\mathcal{T}'\subseteq\mathcal{T},|\mathcal{T}'|=2\}$ are linearly independent.
    %For any $\{n_1,n_2,n_3\} \subseteq [N]$, $\mathrm{span}\{U_{n_1,n_2;j}^*,U_{n_1,n_3;j}^*,U_{n_2,n_3;j}^*\}$ is exactly the same space as $\mathrm{span}\{U_{1,j},U_{2,j},U_{3,j}\}$. In other words,  $U_{n_1,n_2;j}^*$ does not belong to  $\mathrm{span}\{U_{2n_3-1,j}^*,U_{2n_3,j}^*\}$ for any $\{n_1,n_2,n_3\} \subseteq [N]$. Another equivalent way is to say that 
\end{itemize}
\end{itemize}

The existence of the matrix ${\bf H}^*$ is provided in the next lemma.

\begin{lemma}\label{lem:oval}
Given a prime power $q\geq N-1$, let $\mathcal{C}=\{\mathbf{v}_1,\dots,\mathbf{v}_{q+1}\}=\{(1,x,x^2)|x\in \mathbb{F}_q\}\cup \{(0,0,1)\}$. For every $n\in [N]$, let the $(2n-1)^{\mathrm{th}}$ and $2n^{\mathrm{th}}$ rows of $\mathbf{H}^*$ be any basis of the two-dimensional orthogonal complement $\operatorname{span}\{\mathbf{v}_n\}^{\perp}\subseteq \mathbb{F}_q^3$. Then ${\bf H}^*$ satisfies the two conditions above.
\end{lemma}
\begin{IEEEproof}
By the Vandermonde property, any $3$ vectors in $\mathcal{C}$ are linearly independent. For any $\{n_1,n_2\}\subseteq [N]$, $\mathbf{v}_{n_1}, \mathbf{v}_{n_2}$ are linearly independent.  The intersection $\operatorname{span}\{\mathbf{v}_{n_1}\}^\perp \cap \operatorname{span}\{\mathbf{v}_{n_2}\}^\perp$ is the solution space of the linear system
\[
\begin{cases}
\mathbf{v}_{n_1} \cdot \mathbf{x} = 0, \\
\mathbf{v}_{n_2} \cdot \mathbf{x} = 0.
\end{cases}
\]
Since $\mathbf{v}_{n_1}$ and $\mathbf{v}_{n_2}$ are linearly independent, the coefficient matrix $(\mathbf{v}_{n_1}^\top \; \mathbf{v}_{n_2}^\top)$ has rank $2$, and thus the solution space is one-dimensional. For any $\{n_1,n_2,n_3\} \subseteq [N]$, since $\mathbf{v}_{n_1}$, $\mathbf{v}_{n_2}$, and $\mathbf{v}_{n_3}$ are linearly independent, we can regard $\{\mathbf{v}_{n_1},\mathbf{v}_{n_2},\mathbf{v}_{n_3}\}$ as a basis of $\mathbb{F}_q^3$. Let $\{\mathbf{u}_1, \mathbf{u}_2, \mathbf{u}_3\}$ be the dual basis of $\{\mathbf{v}_{n_1}, \mathbf{v}_{n_2}, \mathbf{v}_{n_3}\}$, i.e.,  
$\langle \mathbf{v}_{n_i}, \mathbf{u}_j \rangle$ equals 1 if $i=j$ and 0 otherwise. 
Then  
$\operatorname{span}\{\mathbf{v}_{n_i}\}^{\perp} = \operatorname{span}\{\mathbf{u}_j \mid j \neq i\}$ and thus $\operatorname{span}\{\mathbf{v}_{n_1}\}^{\perp}\cap\operatorname{span}\{\mathbf{v}_{n_2}\}^{\perp}=\mathrm{span}\{\mathbf{u}_3\}, \operatorname{span}\{\mathbf{v}_{n_1}\}^{\perp}\cap\operatorname{span}\{\mathbf{v}_{n_3}\}^{\perp}=\mathrm{span}\{\mathbf{u}_2\}, \operatorname{span}\{\mathbf{v}_{n_2}\}^{\perp}\cap\operatorname{span}\{\mathbf{v}_{n_3}\}^{\perp}=\mathrm{span}\{\mathbf{u}_1\}$. Therefore, the matrix ${\bf H}^*$ satisfies the two conditions as desired.
\end{IEEEproof}

Look back on the example in the previous subsection for $(M,N,T,K)=(2,6,3,3)$. We take  $\mathbf{v}_1=(1 , 0 , 0)$, $\mathbf{v}_2=( 1 , 1 , 1 )$, $\mathbf{v}_3=(1 , 2 , 4 )$, $\mathbf{v}_4=(1 , 3 , 2 )$, $\mathbf{v}_5=( 1 , 4 ,2 )$, and $\mathbf{v}_6=(0 , 0 , 1 )$. Then we get the desired coefficient matrix as shown in that example. 

By the constructions of the query sets $\mathcal{U}_n$'s for the undesired file, the query spaces $\Gamma'_n$'s have the same properties as the $\Gamma_n$'s, and thus the roles of the desired and undesired files are disguised. Let $W_m$ be the desired file and $W_{m^c}$ be the undesired file. Write each of the query sets $\mathcal{V}_n$'s and $\mathcal{U}_n$'s in the form of $(L/N)\times
\tbinom{N}{K}$ matrices, where each matrix $\mathbb{B}(\mathcal{V})$ is the reduced row echelon form of the matrix $\mathcal{V}$. For $n\in[N]$, let $Q_n^{[m]} =(Q_n^{[m]}(W_m), Q_n^{[m]}(W_{m^c}))
=(\mathbb{B} (\mathcal{V}_n),\mathbb{B} (\mathcal{U}_n))$. After receiving the queries, the $n^{\mathrm{th}}$ server applies the query vectors to its stored contents $(W_{mn},\, W_{m^c n})$, producing $L/N$ queried desired symbols $\mathbb{B} (\mathcal{V}_n)W_{mn}$ and $L/N$ queried undesired symbols $\mathbb{B} (\mathcal{U}_n)W_{m^cn}$. 

\subsubsection{The squeezing phase (Exploit the redundant symbols)}
Since each query vector $V_i$ belongs to the query space sent to exactly $K$ servers from an $(N,K)$-MDS system and $S$ is of full rank, $\{\mathcal{V}_nW_{mn}:n\in[N]\}$ can recover $W_m$.  Therefore, the scheme works as long as the user collects $\{\mathbb{B}(\mathcal{V}_n)W_{mn}:n\in[N]\}$ because there is a bijection between $\{\mathbb{B}(\mathcal{V}_n)W_{mn}:n\in[N]\}$ and $\{\mathcal{V}_nW_{mn}:n\in[N]\}$. As for the undesired file,
we would like that  $\{\mathcal{U}_nW_{m^cn}:n\in[N]\}$ have redundancy as much as possible:
\begin{itemize}
    \item For each $i\in[\delta]$, there are $N-K$ redundant symbols among $\{Z_iW_{m^cn}:n\in[N]\}$.
    \item For $1\leq j\leq \gamma$, due to the relations of $\bf G$ and $\bf H$ from Lemma \ref{GRS}, for a GRS-coded system, there are $N-K-2$ redundant symbols among $\{\widetilde{U}_{n, j}W_{m^cn}:n\in[N]\}$.
    \item Recall that we designed the vectors $\{U_{2n-1,j}^*,U_{2n,j}^*:{n\in[N],j\in[\tau]}\}$ based on Lemma \ref{lem:oval}. How much redundancy can we squeeze among $\{{U}_{2n-1, j}^*W_{m^cn}, {U}_{2n, j}^*W_{m^cn}:n\in[N]\}$ for this construction? A more general problem is, how should we design the vectors $\{U_{2n-1,j}^*,U_{2n,j}^*:{n\in[N],j\in[\tau]}\}$ aiming for maximum redundancy? We believe this problem is highly non-trivial. Nevertheless, when $2N>3K$, there are naturally at least $2N-3K$ redundant symbols among $\{{U}_{2n-1, j}^*W_{m^cn}, {U}_{2n, j}^*W_{m^cn}:n\in[N]\}$, since they are all linear combinations of the inner products of the $3$ vectors $\{U_{1,j},U_{2,j},U_{3,j}\}$ and the $K$ systematic message vectors.   
\end{itemize}
Therefore, in total there are at least $\delta(N-K)+\gamma(N-K-2)+\tau\cdot\max\{2N-3K,0\}$ redundant symbols among $\{\mathcal{U}_nW_{m^cn}:n\in[N]\}$. In other words, the queried undesired symbols have dimension at most $I\triangleq L-\delta(N-K)-\gamma(N-K-2)-\tau\cdot\max\{2N-3K,0\}$.

When the $n^{\mathrm{th}}$ server responds to the query vectors, it applies a combination strategy to map  $(\mathbb{B} (\mathcal{V}_n)W_{mn}, \mathbb{B} (\mathcal{U}_n)W_{m^cn})$ into an equal or smaller number of downloaded symbols. The framework is almost the same as what we have presented for the case $T=2$. A full rank matrix $C_n$ of order $L/N$ is applied to transform $\mathbb{B} (\mathcal{V}_n)W_{mn}$ and $\mathbb{B} (\mathcal{U}_n)W_{m^cn}$ into $\{X_{n,1},\dots,X_{n,{L/N}},Y_{n,1},\dots,Y_{n,{L/N}}\}$. Then a combination function $\mathcal{L}_n$ will compress these $2L/N$ symbols into $I_n$ downloaded desired symbols, $I_n$ downloaded undesired symbols, and $L/N-I_n$ paired-up summation symbols. The key is still to make sure that $\sum_{n\in[N]}I_n=I$ and the union of the downloaded undesired symbols should linearly generate all the queried undesired symbols. However, now it seems difficult or impossible to preset the matrices $C_n$'s in a deterministic way\footnote{In~\cite{sun2017private}, the existence proof of the $C_n$'s for $T=2$ relies on the celebrated Schwarz–Zippel lemma. A key ingredient for the proof is that the number of permutations on each query sets is a constant independent of the field size $q$, and then in sufficiently large fields the Schwarz–Zippel lemma can guarantee the existence of the desired matrices.
However, as pointed out also in~\cite{sun2017private}, such a proof no longer works for $T\geq 3$, since now the possibility of each query space also grows with $q$.}. Therefore, we switch to a randomized combination strategy as in~\cite{sun2017private}, by randomly choosing the matrices $C_n$'s over a sufficiently large finite field and allowing the scheme to have $\epsilon$-error, with $\epsilon$ approaching zero when the message size approaches infinity (for interested readers, please refer to~\cite[Lemma 5]{sun2017private} and the contexts). 

Summing up the above, we have the following theorem. 
\begin{theorem}
For $(M,N,T,K)=(2,N,3,K)$, where $N\geq K+2$, if the storage code is a GRS code, then there exists a PIR scheme with rate $R=\frac{N(N-1)(N-2)}{N^3-N^2+NK^2+KN^2-K^3+3K^2-8NK+4K}$ when $2N>3K$, and $R=\frac{N(N-1)(N-2)}{N^3-3N^2+3N^2K-4NK^2-3NK+2K^3+4K}$ when $2N\leq 3K$.
\end{theorem}

\begin{IEEEproof}
The privacy and $\epsilon$-correctness of the scheme has been explained throughout the presentation of the scheme. The rate of the scheme can be computed as $\frac{L}{L+I}$, where $I= L-\delta(N-K)-\gamma(N-K-2)-\tau\cdot\max\{2N-3K,0\}$. Plugging in $L={N\choose K}K$, $\delta={N-3\choose K-3}$, $\gamma={N-3\choose K-1}$, and $\tau={N-3\choose K-2}$ and the theorem follows.
\end{IEEEproof}

\subsection{The general scheme for arbitrary $T$}

When $T$ becomes larger, the disguising phase becomes more complicated, in the sense that we have to analyze the perspective of $T'$ servers for all $1\leq T'\leq T$. To construct the query spaces satisfying privacy requirement, we will need the following tool of exterior products.

\begin{definition}[Exterior Products]
Let $1\leq T'\leq T$. Consider the tensor power $(\mathbb{F}_q^T)^{\otimes T'}$ and let $\mathcal{W}$ be the subspace generated by all elements of the form $\mathbf{u}_1 \otimes \cdots \otimes \mathbf{u}_{T'}$ where $\mathbf{u}_i = \mathbf{u}_j$ for some $1\leq i <j \leq T'$. The $T'$-th exterior power of $\mathbb{F}_q^T$ is defined as $\Lambda^{T'} (\mathbb{F}_q^T) \triangleq (\mathbb{F}_q^T)^{\otimes T'}/\mathcal{W}$.  The canonical projection from $(\mathbb{F}_q^T)^{\otimes T'}$ onto $ \Lambda^{T'}(\mathbb{F}_q^T)$ maps $\mathbf{u}_1 \otimes \cdots \otimes \mathbf{u}_{T'}$ to $\mathbf{u}_1 \wedge \cdots \wedge \mathbf{u}_{T'}$, referred to as the exterior product of the vectors.
\end{definition} 

The exterior product satisfies the following properties (for example, see~\cite{szekeres2004course}):

\begin{enumerate}
    \item \textbf{Multilinearity}: For any scalars $a, b \in \mathbb{F}_q$ and vectors $\mathbf{u}_1, \ldots, \mathbf{u}_{T'}, \mathbf{u}_1'\in \mathbb{F}_q^T$,
        \[
        (a \mathbf{u}_1 + b \mathbf{u}_1') \wedge \mathbf{u}_2 \wedge \cdots \wedge \mathbf{u}_{T'} = a (\mathbf{u}_1 \wedge \mathbf{u}_2 \wedge \cdots \wedge \mathbf{u}_{T'}) + b (\mathbf{u}_1' \wedge \mathbf{u}_2 \wedge \cdots \wedge \mathbf{u}_{T'}).
        \]
    Similarly, linearity holds on each component separately.
    \item \textbf{Antisymmetry}: For any vectors $\mathbf{u}_1, \ldots, \mathbf{u}_{T'} \in \mathbb{F}_q^T$ and any transposition swapping $i$ and $j$,
        \[
        \mathbf{u}_1 \wedge \cdots \wedge \mathbf{u}_i \wedge \cdots \wedge \mathbf{u}_j \wedge \cdots \wedge \mathbf{u}_{T'} = - \left(\mathbf{u}_1 \wedge \cdots \wedge \mathbf{u}_j \wedge \cdots \wedge \mathbf{u}_i \wedge \cdots \wedge \mathbf{u}_{T'}\right).
        \]
    In particular, if any two vectors among $\{\mathbf{u}_1,\dots,\mathbf{u}_{T'}\}$ are equal, then $\mathbf{u}_1 \wedge \cdots \wedge \mathbf{u}_{T'} = 0$.
\end{enumerate}

Based on these two properties, from an arbitrary basis $\{\mathbf{e}_1, \ldots, \mathbf{e}_T\}$ of $\mathbb{F}_q^T$, we can find a basis of $\Lambda^{T'}( \mathbb{F}_q^T)$ given by
\begin{align}
\{ \mathbf{e}_{i_1} \wedge \mathbf{e}_{i_2} \wedge \cdots \wedge \mathbf{e}_{i_{T'}} \mid 1 \leq i_1 < i_2 < \cdots < i_{T'} \leq T \}.\label{generator}
\end{align}
Thus, $\dim \Lambda^{T'} (\mathbb{F}_q^T) = \binom{T}{T'}$, and we have an isomorphism $\phi:\Lambda^{T'} (\mathbb{F}_q^T) \to \mathbb{F}_q^{\binom{T}{T'}}$.  
\begin{lemma}\label{T,T'}
Let $2\leq T'< T$. There exists a collection of $q+1$ subspaces of $\mathbb{F}_q^{\binom{T}{T'}}$, each of dimension $\binom{T-1}{T'-1}$, such that:
\begin{enumerate}
    \item The intersection of any $T'$ distinct subspaces is of dimension one.
    \item For any $T$ distinct subspaces, the set of generators obtained from the intersection of each $T'$-tuple of these $T$ subspaces is linearly independent and thus spans $\mathbb{F}_q^{\binom{T}{T'}}$.
\end{enumerate}
\end{lemma}
\begin{IEEEproof}
  $\Lambda^{T'} (\mathbb{F}_q^T)$  is a $\binom{T}{T'}$-dimensional linear space over $\mathbb{F}_q$, and we have an isomorphism $\phi: \Lambda^{T'}(\mathbb{F}_q^T) \to \mathbb{F}_q^{\binom{T}{T'}}$. We will construct $q+1$ subspaces in the exterior power $\Lambda^{T'} (\mathbb{F}_q^T)$ satisfying all the conditions of the lemma.  

Consider $q+1$ vectors in $\mathbb{F}_q^T$ given by 
\[
\mathcal{C} = \{(1,x,x^2,\dots,x^{T-1}) | x \in \mathbb{F}_q\} \cup \{(0,0,\dots,0,1)\}.
\]
By the Vandermonde property, any $T$ vectors in $\mathcal{C}$ are linearly independent.

For each  $\mathbf{u}\in \mathcal{C}$, define a subspace $\mathcal{W}_\mathbf{u}$ of $\Lambda^{T'}(\mathbb{F}_q^{T})$ by:
\[
\mathcal{W}_\mathbf{u} = \{ \mathbf{u} \wedge \omega \mid \omega \in \Lambda^{T'-1} (\mathbb{F}_q^T) \}.
\]
Consider any basis of $\mathbb{F}_q^T$ which contains $\mathbf{u}$, say, $\{\mathbf{u}, \mathbf{u}_1, \dots, \mathbf{u}_{T-1}\}$. Then a basis of $\mathcal{W}_\mathbf{u}$ is given by
\begin{align}
\{ \mathbf{u} \wedge \mathbf{u}_{i_1} \wedge \cdots \wedge \mathbf{u}_{i_{T'-1}} \mid 1 \leq i_1 < i_2 < \cdots < i_{T'-1} \leq T-1 \}.\label{W_u}
\end{align}
Thus, $\dim{\mathcal{W}_\mathbf{u}}=\binom{T-1}{T'-1}$.
We will show that these $q+1$ subspaces $\{\mathcal{W}_{\mathbf{u}} \mid \mathbf{u} \in \mathcal{C}\}$ satisfy the two conditions in the lemma.

Now let $\mathbf{u}_1, \ldots, \mathbf{u}_{T'}$ be $T'$ distinct vectors of $\mathcal{C}$.  
These vectors are linearly independent. Extend $\{\mathbf{u}_1, \ldots, \mathbf{u}_{T'}\}$ to a basis 
$\{\mathbf{u}_1, \ldots, \mathbf{u}_{T'}, \mathbf{u}_{T'+1}, \ldots, \mathbf{u}_T\}$ of $\mathbb{F}_q^T$. Let
$\xi \in \bigcap_{i=1}^{T'} \mathcal{W}_{\mathbf{u}_i}$. For each $i \in [T']$, according to \eqref{W_u}, $\xi$ can be linearly generated by the set $
\{ \mathbf{u}_i \wedge \mathbf{u}_{j_1} \wedge \dots \wedge \mathbf{u}_{j_{T'-1}} \mid \forall\, \{j_1, j_2, \dots, j_{T'-1}\} \subseteq [T] \setminus \{i\} \}$. Because $\xi$ has a unique representation with respect to the basis $\{ \mathbf{u}_{i_1} \wedge \mathbf{u}_{i_2} \wedge \cdots \wedge \mathbf{u}_{i_{T'}} \mid 1 \leq i_1 < i_2 < \cdots < i_{T'} \leq T \}$ of $\Lambda^{T'}(\mathbb{F}_q^T)$,
it follows that $\xi$ must be  linearly represented by the single basis vector $\mathbf{u}_{1} \wedge \mathbf{u}_{2} \wedge \cdots \wedge \mathbf{u}_{T'}$, and the intersection $\bigcap_{i=1}^{T'} \mathcal{W}_{\mathbf{u}_i}$ is one-dimensional. 

Next, let $\mathbf{u}_1, \ldots, \mathbf{u}_T$ be $T$ distinct vectors of $\mathcal{C}$.  
They form a basis of $\mathbb{F}_q^T$.  Consider the generators of the intersection of any $T'$ of the $T$ subspaces $\mathcal{W}_{\mathbf{u}_1}, \dots, \mathcal{W}_{\mathbf{u}_T}$.
For every $T'$-subset $S = \{i_1,\ldots,i_{T'}\} \subseteq \{1,\ldots,T\}$, $
\bigcap_{i \in S} \mathcal{W}_{\mathbf{u}_i}
=
\operatorname{span}\{\mathbf{u}_{i_1} \wedge \cdots \wedge \mathbf{u}_{i_{T'}}\}.$ According to \eqref{generator},
as $S$ ranges over all $T'$-subsets, $\{ \mathbf{u}_{i_1} \wedge \mathbf{u}_{i_2} \wedge \cdots \wedge \mathbf{u}_{i_{T'}} \mid 1 \leq i_1 < i_2 < \cdots < i_{T'} \leq T \}$ form the exterior basis of $\Lambda^{T'}(\mathbb{F}_q^T)$, and thus are linearly independent and span $\Lambda^{T'}(\mathbb{F}_q^T)$.

Via the isomorphism $\phi$, the corresponding subspaces $\{\phi(\mathcal{W}_\mathbf{u})|\mathbf{u}\in \mathcal{C}\}$ in $\mathbb{F}_q^{\binom{T}{T'}}$ also satisfy the conditions of the lemma. This completes the proof.
\end{IEEEproof}

With the tool of exterior product, we are now ready to present a disguise-and-squeeze MDS-TPIR scheme for arbitrary $T$. The settings of the files and the storage code are the same as in the previous subsection for $T=3$.

\subsubsection{The disguising phase (Constructions of queries)} Independently choose two random matrices $S, S'$ uniformly from all the  $\tbinom{N}{K}\times \tbinom{N}{K}$ full rank matrices over $\mathbb{F}_q$. Let $\delta_{T'}\triangleq\tbinom{N-T}{K-T'}$, where $1\leq T'\leq T$ (if $K-T'>N-T$ or $K-T'<0$, then $\delta_{T'}$ is set as $0$ by default). 
Label the rows of $S$ as $\left\{V_i: i\in \left[\tbinom{N}{K}\right]\right\}$. Divide the rows of $S'$ into $T$ parts  $\left\{Z_i:i\in [\delta_T]\right\}$, $\bigcup\limits_{1<T'<T}\left\{U_{T';1,j},\dots,U_{T';\binom{T}{T'},j}:j\in[\delta_{T'}]\right\}$, and $\left\{U_\ell:\ell\in \left[{N \choose K}-\sum\limits_{1<T'\leq T}\delta_{T'}\cdot\binom{T}{T'}\right]\right\}$.

The query sets $\mathcal{V}_1,\mathcal{V}_2,\dots,\mathcal{V}_N$ for the desired file are defined in the same way as before, i.e., each $\mathcal{V}_n$ contains $\tbinom{N-1}{K-1}$ row vectors from $S$, and every $K$ sets out of $\mathcal{V}_1,\mathcal{V}_2,\dots,\mathcal{V}_N$ share exactly one common row vector from $S$. Let $\Gamma_{n}$ be the linear space spanned by the query vectors in $\mathcal{V}_{n}$, $n\in [N]$. Then for any $\mathcal{T}=\{n_1,n_2,\dots,n_{T}\}\subseteq[N]$, it holds that $\mathrm{dim}(\bigcap_{t\in [T]}\Gamma_{n_t})=\delta_T$. For $1\leq T'<T$, any $\mathcal{T}'\subseteq \mathcal{T}$ with $|\mathcal{T}'|=T'$, it holds that $\mathrm{dim}\left(\bigcap_{i\in \mathcal{T}'}\Gamma_{i}\setminus\sum_{j\in\mathcal{T}\setminus\mathcal{T}'}\Gamma_{j}\right)=\delta_{T'}$. 
For disguising, these properties should be imitated in the queries regarding the undesired file. Towards this goal,
we define the query sets $\mathcal{U}_1,\mathcal{U}_2,\dots,\mathcal{U}_N$ for the undesired file as follows. For $n\in [1:N]$,
\[
  \mathcal{U}_{n}= \left\{U_{T';i,j}^{*} \middle| 1<T'<T, i\in \left[(n-1)\binom{T-1}{T'-1}+1:n\binom{T-1}{T'-1}\right], j\in [\delta_{T'}]\right\}\cup\{Z_{1} , \dots , Z_{\delta_T}\}\cup 
\{\widetilde{U}_{n,1}, \dots,\widetilde{U}_{n, \delta_1}\}. 
\]
Let $\Gamma'_{n}$ be the linear space spanned by the vectors in $\mathcal{U}_{n}$, $n\in [N]$. The roles of these vectors are as follows.
\begin{itemize}
    \item $\{Z_1,\dots,Z_{\delta_T}\}$ are directly picked from $S'$, accounting for the $\delta_T$-dimensional intersection for any $T$ spaces. 
    \item The vectors $\{\widetilde{U}_{n,j}:{n\in[N],j\in[\delta_1]}\}$ are defined by
\begin{eqnarray*}
\left(\begin{array}{ccc}\widetilde{U}_{1,1} & \cdots & \widetilde{U}_{1, \delta_1} \\ 
\widetilde{U}_{2,1} & \cdots & \widetilde{U}_{2, \delta_1} \\ 
\vdots &  \ddots& \vdots \\
\widetilde{U}_{N, 1} & \cdots & \widetilde{U}_{N, \delta_1 }\end{array}\right)
={\bf H}\left(\begin{array}{lll}U_{1} & \cdots & U_{\delta_1} \\
U_{\delta_1 +1} & \cdots & U_{2\delta_1}\\
\vdots &  \ddots& \vdots \\
U_{(T-1)\delta_1+1}& \cdots & U_{T\delta_1}
\end{array}\right),
\end{eqnarray*} 
where ${\bf H}^{\top}$ is the generator matrix of $\mathrm{GRS}_{T}(\boldsymbol{\alpha}, \bf{w})$ and this equation is well defined since it is routine to check that ${N \choose K}-\sum\limits_{1<T'\leq T}\delta_{T'}\cdot\binom{T}{T'}\geq T\delta_1$.
    \item For each $T'\in [2:T-1]$, the vectors $\left\{U_{T';i,j}^{*} \middle| i\in \left[1:\binom{T-1}{T'-1}N \right],j\in [\delta_{T'}]\right\}$ are used to guarantee that for any $\mathcal{T}=\{n_1,n_2,\dots,n_{T}\}\subseteq[N]$ and $\mathcal{T}'\subseteq \mathcal{T}$ with $|\mathcal{T}'|=T'$, it holds that $\mathrm{dim}\left(\bigcap_{i\in \mathcal{T}'}\Gamma_{i}\setminus\sum_{j\in\mathcal{T}\setminus\mathcal{T}'}\Gamma_{j}\right)=\delta_{T'}$. For any $j\in [\delta_{T'}]$, we define
\begin{eqnarray*}
    \left(U_{T';1,j}^* ,U_{T';2,j}^*, \dots, U_{T';\binom{T-1}{T'-1}N,j}^*\right)^{\top}={\bf H}_{T'}^*\left(\begin{array}{c}
     U_{T';1,j}  \\
     U_{T';2,j}\\
     \vdots\\
     U_{T';\binom{T}{T'},j}
\end{array}\right),
\end{eqnarray*}
where the matrix ${\bf H}_{T'}^*$ is of size $\binom{T-1}{T'-1}N\times \binom{T}{T'}$ satisfying the following conditions: 
\begin{itemize}
     \item[{P1:}] For any $\mathcal{T}'=\{n_1,n_2,\dots,n_{T'}\}\subseteq[N]$, 
     $\bigcap_{t\in[T']}\mathrm{span}\left\{U_{T';i,j}^* \middle| i\in\left[(n_t-1)\binom{T-1}{T'-1}+1:n_t\binom{T-1}{T'-1}\right] \right\}=\mathrm{span}\{\mathbf{u}_{\mathcal{T}';j}\}$ has dimension 1. 
     %and we denote this dimension as $U_{n_1,n_2;j}^*$.
    \item[{ P2:}] For any $\mathcal{T}=\{n_1,n_2,\dots,n_T\} \subseteq [N]$, $\{\mathbf{u}_{\mathcal{T}';j} |\mathcal{T}'\subseteq\mathcal{T},|\mathcal{T}'|=T'\}$ are linearly independent.
    %For any $\{n_1,n_2,n_3\} \subseteq [N]$, $\mathrm{span}\{U_{n_1,n_2;j}^*,U_{n_1,n_3;j}^*,U_{n_2,n_3;j}^*\}$ is exactly the same space as $\mathrm{span}\{U_{1,j},U_{2,j},U_{3,j}\}$. In other words,  $U_{n_1,n_2;j}^*$ does not belong to  $\mathrm{span}\{U_{2n_3-1,j}^*,U_{2n_3,j}^*\}$ for any $\{n_1,n_2,n_3\} \subseteq [N]$. Another equivalent way is to say that 
\end{itemize}    
\end{itemize}

The existence of the matrix ${\bf H}_{T'}^*$ is provided in the Lemma \ref{T,T'} via the tool of exterior products, when $q+1\geq N$.

By the constructions of the query sets $\mathcal{U}_n$'s for the undesired file, the query spaces $\Gamma'_n$'s have the same properties as the $\Gamma_n$'s, and thus the roles of the desired and undesired files are disguised. Let $W_m$ be the desired file and $W_{m^c}$ be the undesired file. Write each of the query sets $\mathcal{V}_n$'s and $\mathcal{U}_n$'s in the form of $(L/N)\times
\tbinom{N}{K}$ matrices, where each matrix $\mathbb{B}(\mathcal{V})$ is the reduced row echelon form of the matrix $\mathcal{V}$. For $n\in[N]$, let $Q_n^{[m]} =(Q_n^{[m]}(W_m), Q_n^{[m]}(W_{m^c}))
=(\mathbb{B} (\mathcal{V}_n),\mathbb{B} (\mathcal{U}_n))$. After receiving the queries, the $n^{\mathrm{th}}$ server applies the query vectors to its stored contents $(W_{mn},\, W_{m^c n})$, producing $L/N$ queried desired symbols $\mathbb{B} (\mathcal{V}_n)W_{mn}$ and $L/N$ queried undesired symbols $\mathbb{B} (\mathcal{U}_n)W_{m^cn}$. 

\subsubsection{The squeezing phase (Exploit the redundant symbols)}

The squeezing phase is almost the same as what we have presented for $T=3$ in the previous subsection. In particular, we also have to settle for a randomized combination strategy, allowing the scheme to have $\epsilon$-error. 

For $T\geq 3$, the redundancy among the the queried undesired symbols $\{\mathcal{U}_nW_{m^cn}:n\in[N]\}$ is calculated as follows. 
\begin{itemize}
    \item For each $i\in[\delta_T]$, there are $N-K$ redundant symbols among $\{Z_iW_{m^cn}:n\in[N]\}$. In other words, $\{Z_iW_{m^cn}:n\in[N]\}$ have dimension at most $\lambda_T\triangleq K$ for each $i\in[\delta_T]$.
    \item For each $j\in [\delta_1]$, due to the relations of $\bf G$ and $\bf H$ from Lemma \ref{GRS}, there are $N-\min\{K+T-1,N\}$ redundant symbols among $\{\widetilde{U}_{n, j}W_{m^cn}:n\in[N]\}$. In other words, for $j\in [\delta_1]$, $\{\widetilde{U}_{n, j}W_{m^cn}:n\in[N]\}$ have dimension at most $\lambda_1\triangleq\min\{K+T-1,N\}$.
    
    \item  Let $W_{m^c}=(\mathbf{x}_1,\mathbf{x}_2,\dots,\mathbf{x}_K)$. For every $T'\in[2:T-1]$ and $j\in [\delta_{T'}]$, we have
\begin{align*}
    \bigcup_{n\in [N]}\left\{U_{T';i,j}^{*}W_{m^cn} \middle| i\in \left[(n-1)\binom{T-1}{T'-1}+1:n\binom{T-1}{T'-1}\right]\right\}\subseteq \mathrm{span}\left\{U_{T';\ell,j}\mathbf{x}_k \middle|\ell\in\left[\binom{T}{T'}\right], k\in [K]\right\},
    \end{align*}
    having dimension at most $\lambda_{T'}\triangleq\min\left\{N\binom{T-1}{T'-1}, K\binom{T}{T'}\right\}$.
\end{itemize}  
Thus, the queried undesired symbols have dimension at most $I\triangleq \sum_{T'\in[T]}\lambda_{T'}\cdot\delta_{T'}$. Summing up the above, we have the following theorem.

\begin{theorem}
    For $(M,N,T,K)=(2,N,T,K)$, $T\geq 3 $, if the storage code is a GRS code, there exists a PIR scheme with rate $R=\frac{\binom{N}{K}K}{\binom{N}{K}K+\sum_{T'\in[T]}\lambda_{T'}\cdot\delta_{T'}}$ where $\delta_{T'}=\tbinom{N-T}{K-T'}$ for $T'\in [T]$, $\lambda_1=\min\{K+T-1,N\}$, and $\lambda_{T'}=\min\left\{N\binom{T-1}{T'-1}, K\binom{T}{T'}\right\}$ for $T'\in [2:T]$.
\end{theorem}

\section{Conclusion}\label{Sec:con}

In this paper we propose a new MDS-TPIR scheme based on a disguise-and-squeeze approach. Our scheme provides more counterexamples for the FGHK conjecture for MDS-coded systems, and achieves linear capacity for GRS-coded systems when $(M,N,T,K)=(2,N,2,2)$ and $N\geq 4$. Moreover, our scheme features a dramatically smaller field size, could be adapted to various extended PIR models, and also could be generalized to $\epsilon$-error MDS-TPIR schemes for $T\geq 3$. We close this paper by listing some main open problems for future research.
\begin{itemize}
    \item  When $K\geq 3$, there always exist non-GRS MDS codes (for example, see~\cite{roth1989construction,beelen2017twisted} and related papers). A general problem is how to aim for maximum redundancy among queried undesired symbols during the squeezing phase for arbitrary non-GRS MDS-coded systems.
    \item Our scheme performs well for two files and also for the $P$-out-of-$M$ multi-file retrieval. Yet, for the single file retrieval with $M\geq 3$, we still have to figure out how to build a general disguise-and-squeeze scheme which could perform better than existing schemes.
    \item For $T\geq 3$, it is expected to analyze possible deterministic combination strategies, such that the $\epsilon$-error schemes can be turned into zero-error schemes.
\end{itemize}

\bibliographystyle{IEEEtran}
\bibliography{arxiv_PIR}
\end{document}